\documentclass[sigconf,nonacm]{aamas} 
\usepackage{todonotes}

\usepackage{balance} %

\title{Project Submission Games in Participatory Budgeting}

\begin{abstract}
  We introduce %
  the framework of \emph{project submission games}, capturing the
  behavior of project proposers in participatory budgeting (and multiwinner
  elections). Here, each proposer submits a subset of project
  proposals, aiming at maximizing the total cost of those that
  get funded.
  We focus on finding conditions under which pure Nash equilibria 
  exist in our games, and on the complexity of checking whether they
  exist. We also seek algorithms for computing best responses for the proposers. Finally, we check the existence of equilibria, as well as their properties, in real-life data.
\end{abstract}


\newcommand{\BibTeX}{\rm B\kern-.05em{\sc i\kern-.025em b}\kern-.08em\TeX}

\usepackage{algorithm}
\usepackage{algorithmic}
\usepackage{mfirstuc}
\usepackage{paralist}

\usepackage{newfloat}
\usepackage{listings}
\DeclareCaptionStyle{ruled}{labelfont=normalfont,labelsep=colon,strut=off} %
\floatstyle{ruled}
\newfloat{listing}{tb}{lst}{}
\floatname{listing}{Listing}

\author{Piotr Faliszewski}
\affiliation{
  \institution{AGH University}
  \city{Kraków}
  \country{Poland}}
\email{faliszew@agh.edu.pl}

\author{Łukasz Janeczko}
\affiliation{
  \institution{AGH University}
  \city{Kraków}
  \country{Poland}}
\email{lukij1997@gmail.com}

\author{Andrzej Kaczmarczyk}
\affiliation{
  \institution{Czech Technical University in Prague}
  \city{Prague}
  \country{Czech Republic}}
\email{andrzej.kaczmarczyk@cvut.cz}

\author{Grzegorz Lisowski}
\affiliation{
  \institution{University of Groningen}
  \city{Groningen}
  \country{The Netherlands}}
\email{g.a.lisowski@rug.nl}

\author{Grzegorz Pierczyński}
\affiliation{
  \institution{University of Warsaw}
  \city{Warszawa}
  \country{Poland}}
\email{g.pierczynski@gmail.com}

\usepackage{booktabs}

\usepackage{multicol}
\usepackage{textcomp}
\usepackage{eurosym}
\usepackage{enumitem}
\usepackage{graphicx}
\usepackage{amsmath}
\usepackage{amsfonts}
\usepackage{amsthm}
\usepackage{xcolor}
\usepackage{amstext}
\usepackage{graphicx}
\usepackage{fancyhdr}
\usepackage{multicol}
\usepackage{verbatim}
\usepackage{float}
\usepackage{thmtools}

\usepackage[utf8]{inputenc}

\usepackage{cleveref}

\theoremstyle{definition}
\newtheorem{definition}{Definition}[section]

\theoremstyle{plain}
\newtheorem{theorem}{Theorem}[section]

\newtheorem{observation}{Observation}[section]
\newtheorem{lemma}[observation]{Lemma}

\newtheorem{example}{Example}[section]

\usepackage{multirow}
\usepackage{xspace}

\usepackage{nicefrac}
\usepackage{verbatim}
\usepackage{tikz}
\usepackage{pgffor}
\definecolor{myblue}{rgb}{0,0,0.6}
\definecolor{myred}{rgb}{0.6,0,0}
\usetikzlibrary{patterns}

\usepackage{relsize}

\tikzset{fontscale/.style = {font=\relsize{#1}}
}
\usepackage{ dsfont }
\usetikzlibrary{backgrounds}
\usetikzlibrary{calc,shapes,decorations,matrix,arrows} 
\usetikzlibrary{fit}

\usepackage{appendix}

\newcommand{\np}{{{\mathrm{NP}}}}
\newcommand{\conp}{{{\mathrm{coNP}}}}

\newcommand{\basicAV}{BasicAV}

\newcommand{\Phragmen}{Phragm\'en}
\newcommand{\Mes}{MES}

\newcommand{\naturals}{\mathbb{N}}

\newcommand{\MesLong}{\ensuremath{\text{Method of Equal Shares}}}

\newcommand{\NE}{{\ensuremath{\text{NE}}}}
\newcommand{\basicAVNE}{\basicAV\text{-}\NE}

\newcommand{\budget}{\ensuremath{B}}
\newcommand{\cost}{\ensuremath{\mathrm{cost}}}

\newcommand{\stratprof}{\ensuremath{\mathbf{s}}}

\usepackage{pifont}%

\newcommand{\xthreec}{\textsc{X3C}}

\renewcommand{\phi}{\varphi}
\renewcommand{\leq}{\leqslant}
\renewcommand{\geq}{\geqslant}

\usepackage{listofitems}
\definecolor{darkred}{rgb}{0.64,0,0}
\definecolor{darkcyan}{rgb}{0,0.55,0.55}
\newcommand{\rowcolor}[1]{\textcolor{black}{#1}}
\newcommand{\columncolor}[1]{\textcolor{black}{#1}}

\newcommand{\nfgame}[1]{%
\setsepchar{ }
\readlist\arg{#1}
\begin{tikzpicture}[scale=0.5]
	\node (RT) at (-2,1) [label=left:\rowcolor{\arg[1]}] {};
\node (RB) at (-2,-1) [label=left:\rowcolor{\arg[2]}] {};
\node (CL) at (-1,2) [label=above:\columncolor{\arg[3]}] {};
\node (CR) at (1,2) [label=above:\columncolor{\arg[4]}] {};
\node (RTL) at (-1.4,0.6) {\rowcolor{\arg[5]}}; %
\node (CTL) at (-0.6,1.4) {\columncolor{\arg[6]}}; %
\node (RBL) at (-1.4,-1.4) {\rowcolor{\arg[7]}};
\node (CBL) at (-0.6,-0.6) {\columncolor{\arg[8]}};
\node (RTR) at (0.6,0.6) {\rowcolor{\arg[9]}};
\node (CTR) at (1.4,1.4) {\columncolor{\arg[10]}};
\node (RBR) at (0.6,-1.4) {\rowcolor{\arg[11]}};
\node (CBR) at (1.4,-0.6) {\columncolor{\arg[12]}};
\draw[-,very thick] (-2,-2) to (2,-2);
\draw[-,very thick] (-2,0) to (2,0);
\draw[-,very thick] (-2,2) to (2,2);
\draw[-,very thick] (-2,-2) to (-2,2);
\draw[-,very thick] (0,-2) to (0,2);
\draw[-,very thick] (2,-2) to (2,2);
\draw[-,very thin] (-2,2) to (0,0);
\draw[-,very thin] (0,0) to (2,-2);
\draw[-,very thin] (-2,0) to (0,-2);
\draw[-,very thin] (0,2) to (2,0);
\end{tikzpicture}}

\newif\ifhideproofs

\ifhideproofs
  \usepackage{environ}
  \NewEnviron{hide}{}

\fi

\makeatletter

\providecommand{\solutionname}{Proof}

\makeatother

\begin{document}

\pagestyle{fancy}
\fancyhead{}

\maketitle

\section{Introduction}\label{sec:intro}

In a typical \emph{participatory budgeting} (PB)
scenario~\citep{cab:j:participatory-budgeting,goe-kri-sak-ait:c:knapsack-voting,rey-mal:t:pb-survey}
a city first fixes some budget, then groups of activists propose
projects that they would like to see implemented---each with its
individual cost and scope---and, finally, the citizens vote on which
of them should be carried out. We focus on the approval setting where
the voters indicate for each project whether they support it or
not, and we study the strategic issues that project proposers face when
deciding which projects to submit.

Let us consider the following example.  The city fixed the budget of
\euro{}$75\,000$ and two activist groups came up with, altogether, six
projects that they could propose, diversified both in terms of their
cost and in terms of their support.  Specifically, Transparent
Revolution in Environmental Engineering group (TREE) developed
projects $T_1$, $T_2$, and $T_3$, focused on building parks in
different areas of the city, whereas the Bicycle Initiative and
Kinetic Enthusiasts society (BIKE) invented projects $B_1, B_2$, and
$B_3$, devoted to improving the cycling infrastructure. The costs and
numbers of voters supporting each of them are given in Table
\ref{table:InitialEx}.

The city decided to use the BasicAV voting rule,\footnote{We note that
  this rule is also called
  GreedyAV~\citep{boe-fal-jan-kac:c:pb-robustness} or
  GreedCost~\citep{rey-mal:t:pb-survey}. We chose to call it BasicAV
  to emphasize that it is the rule most commonly used in practice.}
which considers the projects from the most to the least supported one 
and selects those that fit into the remaining budget at the time of
consideration.
Hence, if both groups submitted all their projects, the rule would
select $T_1$, $B_1$, $B_2$, and $B_3$. However, in order to have as
much of the budget dedicated to their projects as possible, TREE did
not submit $T_1$, as then projects $T_2$, $B_1$, and $T_3$ would be
selected, giving them \euro{}$60\,000$ for their projects, instead of
\euro{}$50\,000$.  However, BIKE realized that this might happen and chose
not to submit $B_1$. Consequently, projects $T_2$, $T_3$, $B_2$ and
$B_3$ were funded. This was fortunate as neither of the groups could
have improved further from this state.  Hence,
choosing what projects to submit involves non-trivial strategic~considerations. %

\begin{table}[h]
	\caption{Costs and levels of support of projects considered by TREE and BIKE groups.}
	\label{table:InitialEx}
	\begin{center}
		\renewcommand{\arraystretch}{1}
		\renewcommand{\tabcolsep}{3pt}
		
		\scalebox{1}{\begin{tabular}{ c l l }
			
                \toprule
			& {Cost (\euro{})} & {Support}   \\ \midrule
		
		${1: T_1 }$ & $50\,000$ & $6\,000$  \\ 
		${2: T_2}$ &  $30\,000$ &   $5\,000$   \\
		${3: B_1 }$ & $10\,000$ & $4\,000$   \\
		
			\bottomrule
		\end{tabular}}%
        \hfill
		\scalebox{1}{\begin{tabular}{ c l l}
			
		\toprule
			& {Cost (\euro{})} & {Support}   \\ \midrule
		
		${4: T_3}$ & $30\,000$ &  $3\,000$  \\
            ${5: B_2}$ & $7\,000$ &  $2\,000$   \\
            ${6: B_3}$ & $7\,000$ &  $1\,000$  \\
		
			\bottomrule
		\end{tabular}}
	\end{center}
\end{table}

Our goal is to provide game-theoretic and computational analysis of scenarios
analogous to those in the above example.
To this end, we introduce \emph{project submission games (PSGs)} where
project proposers are the players, each proposer owns a set of
projects, and their goal is to maximize the amount of money spent on
their projects by choosing which ones to submit. We ask if our games
always admit Nash equilibria and, if not, what is the complexity of
deciding if they exist. We adopt the following assumptions:
\begin{enumerate}%
\item The proposers have full knowledge of the setting. In particular,
  they know exactly which voters approve which projects. %
  Naturally, this is a simplification---in practice, proposers have at most approximate information about voters' preferences, e.g., via polls.  In the theoretical sections, we believe that this restriction is justified, as a first necessary step. 

  We further weaken this %
  assumption in the experimental section (\Cref{sec:experiments-limited-info}). In particular, we show that in realistic instances limited information regarding voters' ballots has little effect on the proposers' utility.

\item The voters approve projects independently from the presence of
  the other projects. This is natural in the context of approval-based
  elections, where voters can simply vote for the projects they
  consider worth implementing (this is in sync with so-called
  \emph{threshold approval
    voting}~\citep{ben-nat-proc-sha:j:pb-elicitation,fai-ben-gal:c:pb-formats}).
\end{enumerate}

We consider several variants of our games, depending on the costs of
the projects, allowed project-submission strategies, and the structure
of the preferences:
\begin{description}
  \item [Costs.]
    By default, we allow arbitrary project costs, but we also analyze
    the setting where all projects have equal, unit costs. This models
    the multiwinner
    setting~\citep{fal-sko-sli-tal:b:multiwinner-voting,lac-sko:b:multiwinner-approval},
    where instead of deciding which projects to fund, the goal is to
    select a committee of a given, fixed size that represents the
    voters. Hence, in this case, the budget is the size of the
    committee and the ``projects'' are the candidates that run for the
    committee (yet, for consistency, we will use terms ``budget'' and ``projects''
    even in this setting).
  \item[Submission Strategies.] By default, each proposer can submit
    any nonempty subset of their projects, but we also consider the
    setting where each proposer submits exactly one project. This
    captures the scenario where the projects are variants of the same
    idea (e.g., with different scope, costs, and targeted for
    different voter groups).

  \item[Preference Structure.] By default, we do not enforce any
    structure over the voters' approvals, but we do consider the
    party-list setting, where each two projects are either approved by
    the same voters or by disjoint groups of voters. This assumption
    most naturally fits the multiwinner framework, where the voters
    vote for ``parties'', but it also makes sense in the PB setting,
    as it generalizes instances where voters approve exactly
    one project, which is sometimes enforced in~practice, e.g., in Wrocław, Poland (see \cite{fal-fli-pet-pie-sko-sto-szu-tal:c:pabulib}).
\end{description}

Overall, we find that for most variants of our games Nash equilibria
are not guaranteed to exist, and it is computationally challenging to
decide if they do. Indeed, this task is often both $\np$- and
$\conp$-hard, and even finding proposers' best responses is
computationally challenging. %
However, for the single-project
submission setting, our games become much easier to analyze,
especially for multiwinner elections. %
Further, our
experiments show that Nash equilibria are common and typically can be
found using simple dynamics. Moreover, %
we demonstrate that the utilities that the proposers obtain under limited information regarding voters' preferences are similar to those in equilibria based on full profiles. So, in the cases in which the proposers only have access to a reliable polls, equilibria can be computed for small subelections, yielding both quick computation and low demand for precise information about the expected votes.

We  emphasize that it is not immediately clear whether the fact that our games are in some contexts hard or easy to play should be considered ``negative'' or ``positive''. Indeed, it is conceivable that the interest of the project proposers (to maximize their profit) might conflict with the interests of the voters or the election designer. We address this issue in \Cref{sec:experiments-sat}.

All missing proofs are available in~\Cref{app:proofs}. We discuss related
work at the end of the paper, in \Cref{sec:related}.

\section{Preliminaries}\label{sec:preli}

A \emph{participatory budgeting election} (an \emph{election}, for
short) is a tuple $E=(C,V, \budget)$ that consists of a set of
\emph{projects}~$C = \{c_1, \dots, c_m \}$, a collection of
\emph{voters} $V= (v_1, \dots, v_n )$, and a \emph{budget}
$\budget \in \mathbb{N}$. Each project $c\in C$ has a \emph{cost}, denoted as
$\cost(c)$, and each voter $v$ has a set of projects
$A(v) \subseteq C$ that they approve. %
For a subset
$W \subseteq C$ of projects, we let
$\cost(W) = \sum_{c \in W}\cost(c)$ be their total cost.  We say that
$W$ is \emph{feasible} if $\cost(W) \leq B$.
For a project $c$, by~$S(c)$ we mean the set of voters that approve it
and we refer to $|S(c)|$ as its \emph{approval score}.
A (participatory budgeting) %
\emph{voting rule}~$f$ is a function that given an election outputs a
feasible set of projects (note that $f$ is resolute, so we assume an internal
tie-breaking mechanism).
In the \emph{multiwinner} setting we assume each project to have cost $1$.

We will further consider plausible domain restrictions. In the \emph{party-list} setting~\citep{brill2018multiwinner} we assume that for each two
projects~$a, b \in C$, it holds that either $S(a) = S(b)$ or
$S(a) \cap S(b) = \emptyset$. In other words, either two project are
supported by the same voters or their sets of supporters are disjoint. Also, the class of \emph{laminar} profiles~\citep{pet-sko:c:welfarism-mes} extends the party-list setting. There, the approval sets of each pair of projects are either disjoint or one of them is a subset of another. Formally, preferences are %
\emph{laminar} if for each pair of projects $a$, $b \in C$ either (i)~$S(a) \subseteq S(b)$, or (ii)~$S(b) \subseteq S(a)$,
or (iii)~$S(a) \cap S(b) = \emptyset$.

\paragraph{Participatory Budgeting Rules.}
For the case of PB, we focus on {\basicAV},~{\Phragmen} \citep{bri-fre-jan-lac:c:phragmen,los-chr-gro:c:phragmen-pb}
and Method of Equal Shares
(\Mes)~\citep{pet-sko:c:welfarism-mes,pet-pie-sko:c:pb-mes}
rules. Each of them starts with an empty set~$W$ and performs a number
of iterations, either extending $W$ with a single project or dropping
some project from consideration. %
We assume that there is a tie-breaking order~$\succ$ over the projects
and whenever a rule faces an internal tie, it selects the project
ranked highest in this order. For an election $E = (C,V,\budget)$, our
rules work as follows: %
\begin{description} %
\item[{\basicAV}.] In each iteration, {\basicAV} considers a
  project~$c$ with the highest approval score that it has not
  considered yet. If $W \cup \{c\}$ is feasible then it includes $c$
  in $W$ (i.e., it \emph{selects} or \emph{funds} $c$) and otherwise it
  drops it. The rule terminates after considering all the projects.

\item[\Phragmen.] %
  Initially, the voters have empty bank accounts, but they
  receive virtual currency at a constant rate of one currency unit per one
  unit of time.
  An iteration starts as soon as there is a project $c$ that has not
  been considered yet, such that voters in $S(c)$ have $\cost(c)$
  currency.
  If $W \cup \{c\}$ is feasible, then {\Phragmen} includes~$c$ in $W$
  and resets the accounts of the voters from $S(c)$ to zero (these
  voters \emph{buy} the project); otherwise it drops~$c$.  The rule
  terminates when no further iteration can start.

\item[\MesLong\ (\Mes).]%
	\!\!\footnote{Several variants of~\Mes{} have been proposed
		in the context of participatory budgeting, depending on the
		presumed utilities of the voters from approved projects. Here we consider
		the \emph{cost-utility} variant, introduced by \citet{fal-fli-pet-pie-sko-sto-szu-tal:c:pabulib}. An alternative definition is discussed in \Cref{app:mes}.}
  Initially, each voter gets~$\nicefrac{\budget}{|V|}$ virtual
  currency. Then, each iteration proceeds as follows (let $b(v)$ be
  the money held by voter $v$ at the beginning of the iteration): For
  a project $c$ and a number $\rho$, we say that $c$ is
  $\rho$-affordable if it holds that:
  \begin{equation*}
    \textstyle
    \sum_{v\in S(c)} \min(b(v), \rho \cdot \cost(c)) = \cost(c)
  \end{equation*}

  i.e.,\ if the voters approving $c$ can afford it, provided that each
  of them contributes $\rho$ fraction of its cost (or all the money
  they have left if the $\rho$ fraction would be too much).  {\Mes}
  includes in $W$ the project that is affordable for the lowest value
  of $\rho$ (and takes
  $\min(b(v), \rho\cdot \cost(c))$ currency from each $v \in S(c)$; these voters \emph{buy} $c$).  {\Mes}
  terminates if no project is affordable for any $\rho$.
\end{description}

{\Mes} is well-known to occasionally output sets of projects that are
not exhaustive, i.e., sets that could be extended with some project
and remain
feasible~\citep{pet-sko:c:welfarism-mes,pet-pie-sko:c:pb-mes}. In
practice, various completion methods are used (see, e.g., the work of
\citet{fal-fli-pet-pie-sko-sto-szu-tal:c:pabulib} for their
comparison), but we disregard this issue for simplicity.

For the multiwinner setting, it is also common to study so-called
Thiele rules~\citep{thi:j:thiele,lac-sko:j:abc-rules} and their
sequential variants. This includes, e.g., the approval-based
Chamberlin--Courant rule (CC) and the proportional approval voting
rule (PAV).  We define (both sequential and global) Thiele rules in
\Cref{app:thiele}, and
we provide results for them there, but we omit them from the main
discussion for the sake of clarity. Indeed, modulo minor technicalities, our
results for them are analogous to those for \Phragmen{} or \Mes.

\section{Project Submission Games}\label{sec:psg}

Let $f$ be a voting rule.  A \emph{project submission game} $G$ for
rule $f$ (an $f$-PSG $G$) consists of (1)~a set of projects
$C=\{c_1, \ldots, c_m\}$, (2)~a set of \emph{project proposers}
$P = \{P_1, \ldots, P_\ell\}$, (3)~a collection of voters
$V=(v_1, \ldots, v_n)$, and (4)~a budget $\budget$.
The components $C$, $V$, $\budget$ are analogous to a PB election.
The proposers are the players in the game and each of them has some
potential projects that they can choose to submit or not. Hence, we
view $P = \{P_1, \ldots, P_\ell\}$ as a partition of $C$.  Thus, we may sometimes speak of a proposer $P_i$, meaning the player
in the game, and of a set of projects $P_i$, meaning the projects that
this proposer may submit. The intention will be clear from the
context. We drop the rule $f$ from the notation if it is clear or
irrelevant.

Formally, a \emph{strategy} $s_i \subseteq P_i$ of proposer $P_i$ is a
nonempty subset of their submitted projects. %
 A~\emph{strategy profile} is a collection of such strategies,
denoted as~$\stratprof = (s_1, \dots, s_{\ell})$.
By $\stratprof_{-i}$ we denote the strategy profile $\stratprof$ excluding the proposer $P_i$, and by
$(\stratprof_{-i}, s_i')$ we denote the strategy profile obtained
from $\stratprof$ by replacing the strategy of $P_i$ with $s_i'$.
An election \emph{induced} by $\stratprof{}$ is a tuple
$E(\stratprof) = (\bigcup_{i \in [\ell]} s_i, V, \budget)$.
The \emph{utility} of a given proposer $P_i$ is defined as the total
cost of the selected projects that they submitted; formally,
$u_i(\stratprof)=\cost(f\big(E(\stratprof)\big) \cap P_i )$.

We analyze the existence, structure, and computational complexity of
computing \emph{pure Nash equilibria} (NE) in project submission
games. Given a game $G$, a strategy profile $\stratprof$ is an NE
profile if for each proposer $P_i$ and all $P_i$'s strategies $s_i' \subseteq P_i$,
we have %
$u_i(\stratprof_{-i}, s_i') \leq u_i(\stratprof)$. 

\begin{example}\label{def:InitialEx}
Take an instance as shown in Section \ref{sec:intro}. There, we have in total six projects that can potentially enter the elections, i.e., $T_1, T_2, T_3, B_1, B_2, B_3$, who are supported by the numbers of voters as shown in Table \ref{table:InitialEx}. Their costs are also as in Table \ref{table:InitialEx}, and the budget is \EUR{75 000}. Moreover, there are 
two proposers, $T$ and $B$, partitioning the projects into sets $T=\{T_1, T_2, T_3\}$ and $B=\{ B_1, B_2, B_3 \}$, and the voting rule used is \basicAV. Consider first the strategy profile $\stratprof = (\{T_1, T_2, T_3,  B_1, B_2, B_3 \})$. There, we see that the funded projects in the  election induced by $\stratprof$~are $T_1, B_1, B_2$ and $B_3$, giving $u_{T}(\stratprof) = 50 000$, and $u_{B}(\stratprof) = 24 000$. Nevertheless, in the profile $\stratprof' = (\{T_2, T_3,  B_1, B_2, B_3 \})$,  $u_{T}(\stratprof) = 60 000$, and $u_{B}(\stratprof) = 10 000$, so \stratprof~is not a \basicAVNE.
\end{example}

In particular, we are 
interested in computing best responses of given
proposers and in checking an NE existence. %
Formally, we consider the following problems.

\begin{definition}
  Let $f$ be a voting rule.  (1)~In the $f$-\textsc{NE Existence} problem
  we are given an $f$-PSG $G$ and we ask if there is an NE profile for
  $G$.
  (2)~In the $f$-\textsc{Best Response}
  problem, we are given an $f$-PSG $G$, a proposer $P_i$, a value $x$,
  and a strategy profile $\stratprof_{-i}$ of players other than
  $P_i$;  we ask if there is a strategy $s_i$ of $P_i$ such that
  $u_i(\stratprof_{-i}, s_i) \geq x$. %
\end{definition}

In the next section we analyze the existence of Nash equilibria and
the complexity of the above problems for PSGs. Then, in \Cref{sec:psg1},
we provide similar analysis for a variant of PSGs where each proposer
submits exactly one project.

\section{Theoretical Analysis}\label{subsec:NEexistence}

Here, we study the conditions under which an NE exists and can be computed efficiently, which turn out to be quite restrictive.
Specifically,
for each of BasicAV, \Phragmen, and \Mes{} there are PSGs without Nash
equilibria even for party-list profiles (and for BasicAV and \Mes{} it
even suffices to consider games with a single voter).

\begin{restatable}{proposition}{propNeNotExist}\label{prop:NeNotExist}
  For each rule among BasicAV, \Phragmen{} and \Mes{}, there is a PSG
 with two proposers that does not
  admit any NE, where the voters have party-list preferences
  (for the case of BasicAV and \Mes{} even with a single voter). %
\end{restatable}
\begin{proof} %
  For \basicAV{} and \Mes, form a game with two
  proposers, $P_1 = \{a_1, a_2, a_3 \}$ and $P_2 = \{b_1, b_2, b_3 \}$
  and the projects' costs:
  \begin{align*}
    \cost(a_1) = 5&,& \cost(a_2) = 4&,& \cost(a_3) = 2,\\
    \cost(b_1) = 6&,& \cost(b_2) = 4&,& \cost(b_3) = 4.
  \end{align*}
  We fix budget $\budget=14$ and we let the tie-breaking order be:
  \[
    a_1 \succ b_1 \succ b_2 \succ b_3 \succ a_2 \succ a_3.
  \]
  There is only one voter, approving all the projects, so both rules
  consider the projects in the tie-breaking order.

  First, observe that $P_1$ always (weakly) benefits from submitting
  $a_2$ and $a_3$. Indeed, if they are not elected, their presence
  does not matter. Further, $a_3$ is the last in tie-breaking, so if
  it is selected, it does not prevent any other (more valuable)
  projects of $P_1$ from being funded. Selecting $a_2$ could only
  possibly prevent $a_3$ from being selected, but
  $\cost(a_2) \geq \cost(a_3)$. By an analogous reasoning, we can see
  that $P_2$ always (weakly) benefits from submitting $b_2, b_3$.

  Now the only actual dilemma of the players is whether to submit
  $a_1$ and $b_1$. The utilities of the players from playing the
  corresponding strategies are summarized in
  \Cref{fig:no-limit-no-ne}. Note that there is no NE in this game
  (indeed, the cycle of best responses is analogous to the one in the classic \emph{matching pennies}
  game).

  \begin{figure}[t]\centering%
    \scalebox{.8}{\nfgame{\Large{$a_1\in s_1$}
        \Large{$a_1\notin s_1$} \Large{$b_1\in s_2$~~~~~~}
    \Large{~~~~~~$b_1\notin s_2$} $7$ $6$ $0$ $14$ $5$ $8$ $6$
        $8$}}
	\caption{Normal form representation, with rows representing the choices of $P_1$ and the columns to the choices of $P_2$.}\label{fig:no-limit-no-ne}
  \end{figure}

  For \Phragmen, exactly the same counterexample does not
  work, as when all the projects have the same support, \Phragmen{}
  prefers cheaper projects (e.g, $a_2$) to more expensive ones (e.g.,
  $a_1$). But here it is enough to consider a plurality election in
  which each project $c$ has a group of $\cost(c)$ supporters. Then
  under \Phragmen{} all the projects are tied, and the reasoning
  presented above works. %
\end{proof}

Worse yet, given a project submission game it is intractable to decide
if it admits a Nash equilibrium.

\begin{restatable}{theorem}{thmnemainnegative}\label{thm:ne-main-negative}
    For each rule among BasicAV, \Phragmen{} and \Mes{},
  \textsc{NE Existence} 
   is both
  $\np$-hard and $\conp$-hard (for BasicAV and \Mes{} this result
  holds even for party-list profiles with a single voter).
\end{restatable}
\begin{proof}
  We give an $\np$-hardness proof for \basicAV{} and \Mes. The result for
  \Phragmen{} will follow from the proof of
  \Cref{thm:multiwinner-hard} later on.  We reduce from the
  \textsc{Subset Sum} problem, whose instance consists of a set of
  positive integers~$U$ and a target value $T$; we ask if there is a
  subset of $U$ that sums up to $T$ (with every number in
  $U$  smaller than $T$). Adapting our reductions, we
  also obtain $\conp$-hardness proofs, available in~\Cref{app:proofs}.

  \textit{Construction:} Our reduction forms a PSG with three
  proposers, $P_0$, $P_1$, $P_2$. The projects of $P_1$ and $P_2$ are
   as in the proof of \Cref{prop:NeNotExist}, whereas $P_0$
  has a single project $c_i$ for each number $x_i \in U$, with
  $\cost(c_i)$ is the $i^{th}$ number in $U$. Then, $P_0$ also has project a $c^*$ with
  cost $15\cdot T-14$.  Then, $\budget = 15 \cdot T$.  There
  is a single voter approving all the projects. The tie-breaking order
  prefers all the projects of $P_0$ to the projects of $P_1$ and
  $P_2$, with $c^*$ being least preferred among the projects of $P_0$,
  and between the projects of $P_1$ and $P_2$ the tie-breaking is as
  in the proof of \Cref{prop:NeNotExist}.
  This construction is the same for both voting rules.

  \textit{Correctness:} Suppose that there exists a solution to the
  considered \textsc{Subset Sum} instance. Then, the best strategy for
  $P_0$ is to submit the projects corresponding to the choice of
  elements from $U$ summing up to $T$. Then the satisfaction of $P_0$
  is $15\cdot T$ (which is the maximal possible) and the satisfaction
  of $P_1$ and $P_2$ is $0$, regardless of their strategies. This
  strategy profile clearly is an NE.

  Suppose now that there does not exist a solution to the considered
  subset sum problem instance. Then the best strategy of $P_0$ is to
  submit only project $c^*$. Note that any other choice of projects
  would give $P_0$ utility at most $15\cdot T-15$. However, then there
  is no NE due to $P_1$ and $P_2$, as shown in the proof of
  \Cref{prop:NeNotExist}.
\end{proof}

This theorem calls for some comments.
First, as we show our problems to be both $\np$-hard and
$\conp$-hard, we expect their exact complexity to be on the second
level of the Polynomial Hierarchy (informally speaking, our hardness
proofs put them beyond the first level, and expressing NE existence in
the form ``there exists a strategy profile such that all deviations
are not profitable'' puts them within the second level).  However, in
the follow-up intractability proofs we generally do not aim to find
both $\np$- and $\conp$-hardness, or exact completeness results, and
focus on the intractability that is more convenient to prove.
Second, it is interesting to ask whether \Phragmen-\textsc{NE
  Existence} also is intractable for party-list profiles. %
  We leave this issue open.
Finally, we note that the proof of \Cref{thm:ne-main-negative} relies
heavily on the projects having different costs. Hence, one might hope
for polynomial-time algorithms for the multiwinner setting, where all
costs are equal. %
However, it turns out that this is the case only for \basicAV.%

\begin{restatable}{theorem}{thmmultiwinnerhard}\label{thm:multiwinner-hard}
  For the multiwinner setting, we have that:
  \begin{enumerate}
    \item Each \basicAV-PSG has a polynomial time computable NE.
    \item %
  \Phragmen-\textsc{NE
        Existence} and \Mes-\textsc{NE
        Existence} %
      are both $\np$-hard and $\conp$-hard .      
  \end{enumerate}  
\end{restatable}
\begin{proof}
Regarding BasicAV, it is enough to note that in the multiwinner
  model it boils down to electing $\budget$ projects with the highest
  number of approvals (up to
  tie-breaking). %
  So, it is clear that submitting all the projects 
  is a dominant strategy of every
  proposer.

  We show $\conp$-hardness of \Phragmen-\textsc{NE Existence}.
  The constructions for the $\np$-hardness case and for \Mes\ are similar and
  are presented in~\Cref{app:proofs}. We reduce from the \textsc{3-SAT} problem,
  where we are given a propositional logic formula in CNF form (with three
  literals per clause) and we ask if it is satisfiable.

  \noindent\textit{The Gadget:} We first show a \Phragmen-PSG for the
  multiwinner setting that does not admit an NE. There are $48$ voters
  and 3 proposers, $P_1 = \{p_1, p_2, p_3 \}$,
  $P_2= \{q_1, q_2, q_3 \}$, and $P_3=\{r_1, r_2\}$. All the projects
  cost $1$, and the budget is $4$. The voters' preferences are
  depicted in \Cref{fig:mes-laminar-none}, with $r_1, r_2$ being
  lowest in the tie-breaking order (the remaining part of the
  tie-breaking order is irrelevant).

\begin{figure}[t]
\centering
  \scalebox{0.5}{\begin{tikzpicture}

  \node at (1, -0.5) {\huge{$6$}};
  \node at (3.5, -0.5) {\huge{$9$}};
  \node at (6.5, -0.5) {\huge{$9$}};
  \node at (9, -0.5) {\huge{$6$}};
  \node at (11, -0.5) {\huge{$6$}};
  \node at (13, -0.5) {\huge{$6$}};
  \node at (15, -0.5) {\huge{$6$}};

    \draw [fill=red!40!white] (0,0) rectangle (12,1);
    \draw [fill=yellow!40!white] (12,0) rectangle (14,1);
    \draw [fill=yellow!40!white] (14,0) rectangle (16,1);

 \draw [fill=blue!40!white] (2,1) rectangle (12,2);

\draw [fill=blue!40!white](2,2) rectangle (5,3);
\draw [fill=blue!40!white] (5,2) rectangle (8,3);

\draw [fill=red!40!white](8,2) rectangle (10,3);
\draw [fill=red!40!white]  (10,2) rectangle (12,3);

    \node at (6, 0.5) {\huge{$p_1$}};
    \node at (7, 1.5) {\huge{$q_1$}};
    \node at (3.5, 2.5) {\huge{$q_2$}};
  \node at (6.5, 2.5) {\huge{$q_3$}};
  \node at (9, 2.5) {\huge{$p_2$}};
  \node at (11, 2.5) {\huge{$p_3$}};
  \node at (13, 0.5) {\huge{$r_1$}};
  \node at (15, 0.5) {\huge{$r_2$}};
  \end{tikzpicture}}
\caption{Gadget for the proof of
  \Cref{thm:multiwinner-hard}. Projects of different proposers are
  depicted as boxes of different colors. Sets of voters of given
  cardinalities support projects shown above them (e.g., 9
  voters approve set
  $\{p_1,q_1,q_2\}$).}\label{fig:mes-laminar-none}
\end{figure}

  Note that $P_3$ always (weakly) benefits from submitting
  $r_1, r_2$, $P_1$ (weakly) benefits from submitting $p_2, p_3$ and
  $P_2$ (weakly) benefits from submitting $q_2, q_3$ (this is so,
  as for each of the proposers, these are the last projects that
  \Phragmen{} would---if at all---select for them, and they are
  approved by disjoint voter sets). Hence, the only actual dilemma is
  whether $P_1$ should submit $p_1$ and if $P_2$ should submit
  $q_1$. In \Cref{app:proofs}, we show that there is no NE in this game
  with respect to these two strategies of $P_1$ and $P_2$, and that if
  we removed two voters supporting $p_2$ and $p_3$ respectively, there
  there would be an NE for $P_1$ submitting $p_1$ and $P_2$ not
  submitting $q_1$.

  \noindent\textit{Construction:} Consider a \textsc{3-SAT} formula $\varphi$
  with $k$ clauses and $w$ variables.  We form a \Phragmen-PSG as
  follows. First, we include a copy of the above-described
  gadget. Second, we add $3w+k+6$ new projects (each with cost $1$):
  (1) For each variable $x$, we form two \emph{literal projects}
    $c_{x}, c_{\neg x}$.
 (2) For each clause, we form a corresponding clause project.
  (3) We include one \emph{extra project} $c^*$.
(4) We include $w+5$ \emph{artificial projects}.
  Tie-breaking is as follows (the exact order within project sets is
  irrelevant for the proof):
  \[
  \substack{\mathit{literal}\\ \mathit{projects}} \succ
  \substack{\mathit{clause}\\ \mathit{projects}} \succ c^* \succ
  \substack{\mathit{gadget}\\ \mathit{projects}} \succ
  \substack{\mathit{artificial}\\ \mathit{projects.}}
  \]
  In addition to the proposers from the gadget, we also have:
  
(1) Proposer $L$, who owns all the literal projects and $c^*$,
(2) Proposer $K$, who owns all the clause projects, and
(3) Proposer $R$, who owns all the artificial projects.
  
  So, we have $48+13w+36$ voters, partitioned as follows:
  \begin{enumerate}
  \item We have $48$ voters from the gadget, with approval sets as in
    the gadget, except that one of the voters approving $p_2$ and one
    supporter of $p_3$ also approve all the clause
    projects.
  \item For each variable $x$, we have $7$ voters, with $5$ of them
    approving both $c_x, c_{\neg x}$, one voter, called $v_x$,
    approving $c_x$, and one voter, called $v_{\neg x}$, approving
    $c_{\neg x}$. Additionally, $v_x$ and $v_{\neg x}$ approve all the
    clause projects that correspond to clauses containing
    corresponding literals.  
  \item Each artificial project has  $6$ supporting voters.
  \item We have $6$ voters approving $c^*$. One of them also approves
    all the clause projects.
\end{enumerate}

  Importantly, each non-gadget project is approved by a group of exactly
  $6$ voters. We set budget $\budget=w+5$.

  \noindent\textit{Correctness:} First, observe that
  each artificial project is approved by a disjoint group of voters who
  do not approve any of the other projects. So, it is always
  (weakly) profitable for $R$ to submit all of them. Hence, if the
  game has an NE, it has one where this happens, and we only consider
  strategy profiles of this form.  As there are $\budget$ artificial
  projects and their voters can purchase them at moment
  $\nicefrac{1}{6}$, if some other project becomes affordable at a
  later time, it is not funded. This means that each
  non-gadget project $c$ (and also gadget projects $p_2, p_3$) can
  only be selected if no other project $c'$ whose set of approving
  voters intersects with that of $c$ is selected.

  From the above observation, we immediately obtain that the following
  pairs of projects cannot be bought together: (i) any two clause
  projects, (ii) two literal projects corresponding to the same
  variable, (iii) any clause project and $c^*$, (iv) any clause project
  and either $p_2$ or $p_3$, (v) a clause project and a literal
  project such that the corresponding literal appears in the
  corresponding clause.

  However, for each pair of literals corresponding to
  different variables, one is guaranteed to be selected, provided that
  $L$ submits it (this holds because literal projects are preferred in
  the tie-breaking order and are supported by six voters each; only
  $p_1$, $q_1$, $q_2$ and $q_3$ gadget projects are supported by a
  larger number of voters).

  Suppose that there exists a valuation of variables satisfying
  $\varphi$. Then there is no NE in the constructed game. Indeed, in a
  strategy maximizing the utility of proposer $L$, they submit
  the projects corresponding to the literals satisfying $\varphi$, as
  well as $c^*$. As literal projects are preferred to the clause
  ones in tie-breaking, their voters buy them before any of the clause
  projects can be purchased. Then, as per item (v) above and
  independently of the strategy of $K$, no clause project can ever
  selected. Then $c^*$ is bought, giving $w+1$ utility to $L$ (which
  is the maximal possible utility that $L$ can achieve). Note that if
  $K$ submitted all their clause projects and $L$ played any other
  strategy, not corresponding to the valuation satisfying $\varphi$,
  then at least one clause project---corresponding to the not
  satisfied clause---would be elected, after which $c^*$ would not be
  elected (hence $L$ gets only $w$ utility). Finally, there
  are $4$ seats in the committee that we did not discuss yet. Since no
  clause project is bought, the gadget is separated from the remaining
  voters, and the analysis in~\Cref{app:proofs} shows that there is
  no NE regarding $P_1$ and~$P_2$.

  Suppose now that the $\varphi$ is not satisfiable. Then let $L$
  submit any $w$ literal projects corresponding to different
  variables, $P_1$ submit $\{p_1, p_2, p_3\}$, $P_2$ submit
  $\{q_2, q_3\}$ and $P_3$, $R$, and $K$ submit all their
  projects. Then \Phragmen{} selects successively the projects of
  $P_1$ and $P_2$ after which it selects the literal projects
  submitted by $L$, one of the clause projects submitted by $K$
  (corresponding to a non-satisfied clause) and the projects of
  $P_3$. Since $\varphi$ is not satisfiable, it is impossible for $L$
  to ``block" the election of all the clause projects, which means
  that it is impossible for $c^*$, $p_2$, and $p_3$ to be selected. It
  is then clear that none of the proposers can deviate to improve
  their utility. Hence, this is an NE.
\end{proof}

Nevertheless, the multiwinner setting guarantees the existence of
equilibria when voters' preferences are party-list.

\begin{theorem}\label{thm:ne-phragmen-mes-party-laminar}
  Every \Phragmen{}- and \Mes{}-PSG,  
for the multiwinner setting with party-list preferences, admits a Nash equilibrium.
\end{theorem}
\begin{proof}

  Let all the proposers submit all their projects. We now show that
  such a strategy profile $\stratprof$ is an NE. Suppose that it is
  not the case. Then there exists a proposer $P_i$ who has an
  incentive to withdraw some of their projects. Suppose that
  $P_i$ performs such a withdrawal one-by-one. If $P_i$ withdraws a
  project $c$ that is not selected, the outcome and their utility
  does not change. However, if $P_i$ withdraws a selected project $c$,
  then the only change in the outcome is the removal of $c$ and either
  (1) adding some other project $c'$ from the same ``party'' (i.e., a
  project supported by the same voters), or (2) adding some other
  project $c'$ from another ``party'' (only for \Phragmen{}), or (3) electing an outcome that has one project fewer (only
  for \Mes). In all three cases, $P_i$ either does not gain anything
  (if $c'$ is a project of $P_i$) or loses utility (if $c'$ belongs to
  another proposer, or the rule selects fewer projects).
\end{proof}

The requirement for the preferences to be party-list in \Cref{thm:ne-phragmen-mes-party-laminar} is quite strict---for example, there exist multiwinner \Mes{}- and \Phragmen-PSGs with no NE where voters' preferences belong to a more general class of laminar preferences, as in the case of PSGs used as gadgets in the proof of \Cref{thm:multiwinner-hard}. 

In principle, even if testing whether a Nash equilibrium exists is
hard, it may still be possible to efficiently compute a best response
for a given strategy profile. However, in case of PSGs the tractability of computing best responses behaves in the same way as in the case of computing Nash equilibria (which mostly follows from respective proofs for \textsc{NE Existence}).

\begin{restatable}{theorem}{thmbestresponse}\label{thm:best-response}
  \textsc{Best Response} is (1) $\np$-hard for BasicAV, \Phragmen{}, and \Mes,  in the general PB model, even with party-list preferences,
       (2) polynomial-time solvable for \basicAV\ in the multiwinner model, (3) $\np$-hard for \Mes\ and \Phragmen\ in the multiwinner model,
        (4) polynomial-time solvable for \Mes~and \Phragmen~in the multiwinner model with party-list preferences.
\end{restatable}

\section{Single-Project Submission Games}\label{sec:psg1}

So far, we have assumed that each proposer submits any nonempty
subset of their projects. However, sometimes this is not realistic.
For example, a proposer may have several different variants of a
project, such as a bike path that goes through several overlapping
routes. Another possibility is that, in a multiwinner election, it may be customary---or
enforced by the rules---that a single institution nominates only one
candidate. Hence, in this section we focus on games that capture such
scenarios.

A \emph{single-project submission game} for rule $f$
($f$-PSG/1) is defined in the same way as a regular $f$-PSG, except
that each proposer submits exactly one
project. %
Further, we refer to variants of $f$-\textsc{NE Existence} and
$f$-\textsc{Best Response} for PSG/1s as, respectively,
$f$-\textsc{Single-Project NE Existence} and
$f$-\textsc{Single-Project Best Response}.

PSG/1s are notably simpler than our full-fledged PSGs. In particular,
in the multiwinner setting not only Nash equilibria always exist for all the rules, but
also they can be computed in polynomial time. We believe that this
setting is quite realistic and, hence, we consider this result
important.\footnote{E.g., in the academic system of the country where
  some of the authors of this paper come from, every few years there
  is an election of a Research Excellence Council that decides
  professors' promotions. For each academic discipline, each
  university department may nominate a candidate and BasicAV is used
  to select three council members for this discipline. Consequently,
  the departments play a PSG/1, as they need to consider the actions
  of other departments.}

\begin{theorem}\label{thm:ne-exists-for-limit-1}
  For BasicAV, \Phragmen{}, and \Mes{} %
  every PSG/1 for the multiwinner setting has a polynomial-time
  computable
  NE. %
\end{theorem}
\begin{proof}
  Note that all our rules are sequential and satisfy the
  following condition: If $W$ is the outcome computed in the first $i$
  rounds and $c$ is a project not in $W$, then if we removed $c$
  from the election and rerun the rule, the first $i$ rounds would
  select the same projects as with $c$ present.

  Suppose for a while that there is no limit on the number of submitted projects and let the
  proposers submit all their projects. Consider the first round. Here,
  project $c_1$ of some proposer $P_1$ is selected. By the above
  observation, we know that $c_1$ would also be selected if it were
  the only project submitted by $P_1$. As no proposer can obtain
  utility higher than $1$ in a multiwinner PSG/1, submitting $c_1$ is
  a dominant strategy for $P_1$. Let us then fix the strategy of $P_1$ as $\{c_1\}$, remove the remaining projects of $P_1$ from the election and rerun the rule. From our observation, $c_1$ is still elected in the first round.
  We now repeat analogous reasoning for projects chosen in further rounds, each time
  finding a dominant strategy of a new proposer $P_i$ (subject to the
  fixed choices of the proposers %
  who got their projects selected earlier). In the end, we choose arbitrary legal strategies for the
  proposers who do not have any projects selected. The presented algorithm returns an NE profile in polynomial time. %

\end{proof}

However, this result does hold in the general PB case.%

\begin{restatable}{proposition}{propNeNotExistLimit}\label{prop:NeNotExistLimit}For each rule among BasicAV, \Phragmen{} and \Mes{}, there is a PSG/1
 with two proposers without NE, assuming party-list preferences
  (for the case of BasicAV and \Mes{} even with a single voter).

\end{restatable}

Nevertheless, in the multiwinner setting, restricting proposers' actions to submitting only a
single project allows for computing best responses in polynomial time,
by iterating over all possible strategies of a given player.

\begin{theorem}\label{thm:pb-best-response}
  For each rule among BasicAV, \Phragmen{}, and \Mes,
   \textsc{Single-Project Best
    Response} is solvable in polynomial time.
\end{theorem}

While the polynomial-time algorithm from the above theorem is simply a
form of brute-force
search, %
it is important to have it. In particular, it allows us to test if a
given strategy profile is an NE, and it may help in designing
heuristics or dynamics that try to find equilibria when they exist.

\section{Empirical Analysis}\label{sec:experiments}

Having analyzed the properties of project submission games on the theoretical level, we study the existence of equilibria empirically, focusing on real-life election data. We further provide results on social welfare in Nash equilibria. We conclude with an assessment of how limited information can affect proposers' utility.

\subsection{Equilibrium Existence}\label{sec:EquilibriumExistence}

Our theoretical analysis shows that, in general, Nash equilibria are
not guaranteed to exist neither for PSGs nor PSG/1s.  In this section
we explore how this situation looks in practice.

We considered 187 instances of real-life PB elections with up to 10
projects, taken from the Pabulib dataset
\citep{fal-fli-pet-pie-sko-sto-szu-tal:c:pabulib}. We have chosen this
limit on the number of projects due to the computational complexity of
dealing with our games. %
For each instance, we analyzed cases with $\ell\in\{2,3,4,5\}$ proposers, each of whom
has a roughly equal number of projects (that is, either
$\lfloor \nicefrac{m}{\ell}\rfloor$ or
$\lceil \nicefrac{m}{\ell}\rceil$ projects), assigned to them
uniformly at random. When the number of projects in an election is
smaller than $\ell+1$, then we omit such an instance. Finally, for each non-discarded  instance and number~$\ell$ of proposers, we constructed three PSGs, one for each rule out of \basicAV, \Phragmen, and \Mes. As a result, we obtained $8124$ PSGs (it is not simply $187 \cdot 4 \cdot 3$ due to the omitted instances).

The goal of our experiments is to check how typical it is for our
games to have equilibria and how difficult they are to compute. Towards
this end, for each of our PSGs
we executed the
following protocol, in which each successive step corresponds to a
more involved way of finding an equilibrium:
(1) We first check whether the strategy profile in which all the
  proposers submit all their projects is an NE---henceforth, we refer
  to such profiles where everyone submits all the projects as
  \emph{full strategy profiles};
  
(2) If the full strategy profile is not an NE, we apply the
   best-response dynamics: Starting from the full strategy
  profile, all the proposers simultaneously and independently change
  their strategies to their best responses. We stop
  either when an NE is found, or after 10 iterations without reaching
  an NE (or if we detect a cycle);

(3) If the above dynamics do not end in an NE, then we test if an
  equilibrium exists using brute-force search.
For PSG/1s, the experimental setup is analogous, except there is no
clear analog of the full strategy profile. Hence we omit the first
step and we use best-response dynamics that starts from a strategy
profile chosen uniformly at random. For \basicAV, we present our results in
\Cref{table:exp-ne-exist-basicav}. The results for other rules are presented
in~\Cref{app:omitted-experiments}. Overall, they are similar to the ones for
\basicAV.

Somewhat surprisingly,
it turns out that not only are Nash equilibria very common in our games, but also they can be computed in a fairly simple
way. Specifically, for a vast majority of our PSGs already the full
strategy profile is an equilibrium. Our dynamics provide an equilibrium for nearly all of the remaining PSGs. There were only five instances, all of which were \Phragmen-PSGs, where to find an equilibrium we had to resort to the brute-force search. Additionally, for PSG/1s our dynamics always found an equilibrium whenever at least one existed.
Notably, we actually never had to exceed 4 iterations of the best response dynamics in order to find an equilibrium.
Moreover, out of all the analyzed games in which an equilibrium was found by invoking our dynamics, for $87\%$ of them it was sufficient to run at most one iteration of the dynamics.\footnote{In total, we used the dynamics to find an equilibirum in $2242$~games. In some PSG/1s games the randomly drawn starting strategy profile was already an equilibrium.}

Hence, our experiments suggest that simply
submitting all the projects is a good basic strategy, but occasionally
the proposers may benefit by considering more involved approaches, based on computing their best responses. As for
PSG/1s we can compute the best responses in polynomial time, this
means that---from practical point of view---finding equilibria for
these games is rather easy.

\begin{table}[t]
  \caption{NE-existence results for \basicAV-PSGs
    (top) and \basicAV-PSG/1s
    (bottom). The column ``All'' shows the numbers of
    considered instances. For each instance, we recorded if the
    respective full strategy profile was an NE (``Full-NE''), if we
    obtained an equlibrium via the step of our algorithm that applies the
    best-response dynamics (``BR-NE''), or if there was no
    equilibrium (``No NE''). We never had to use the brute-force search for
    \basicAV-PSGs or \basicAV-PSG/1s.}
	\label{table:exp-ne-exist-basicav}
	\begin{center}
	\scalebox{.9}{	\begin{tabular}{ c c c  c c}
			
                  \multicolumn{5}{c}{Project Submission Games (PSGs)}\\[1mm]
                  \toprule
		$\ell$	& {All} & {Full-NE} & {BR-NE} & {No NE} \\ \midrule
		
		$2$ & $187$ & $133$ & $48$ & $6$  \\ 
		$3$ & $179$ & $146$ & $26$ & $7$  \\
	  $4$ & $171$ & $149$ & $18$ & $4$ \\
		$5$ & $140$ & $130$ & $9$ & $1$\\
                  \bottomrule
                  \\[-2.5pt]
                  
                  \multicolumn{5}{c}{Single-Project Submission Games (PSG/1s)}\\[1mm]
                  \toprule
        $\ell$	& {All} & {} & {BR-NE} & {No NE} \\ \midrule
		
		$2$ & $187$  &  & $177$  &  $10$  \\ 
		$3$ & $179$  &  & $170$  &  $9$ \\
	  $4$ & $171$  &  & $164$  &  $7$\\
    $5$ & $140$  &  & $136$  &  $4$\\
    \bottomrule
		\end{tabular}}
	\end{center}
\end{table}

\subsection{Voter Satisfaction in Equilibrium}\label{sec:experiments-sat}

We now analyze how the impact of proposers' strategic behavior on the satisfaction of the voters under different voting rules. For an outcome $W\subseteq C$, we measure the voters' satisfaction in two ways:

(1) \emph{Chamberlin-Courant (CC) score:} the number of voters $i$ such that $|A_i \cap W| > 0$, that is, the number of voters who approve at least one selected project.
 (2) \emph{Social Welfare (SW):} for each voter $i$, we define the utility of $i$ from $W$ as $\cost(A_i \cap W)$. Then, SW is defined as the sum of all voters' utilities: $\sum_{i\in N} \cost(A_i \cap W)$.

Both metrics are standard to the social choice literature, measuring the voters' satisfaction in an egalitarian and utilitarian way, respectively. 
We consider the effect that strategizing has on both metrics. Denote as $\stratprof^*$ the strategy profile in which all of the proposers nominate all of their projects (which can be viewed as the most natural, `honest' strategy). Then, for a voting rule $f$ and a strategy profile $\stratprof$, we define the \emph{price of deviation} as 
$
\frac{\text{sat}(\stratprof) - \text{sat}(\stratprof^*)}{\max
(\text{sat}(\stratprof), \text{sat}(\stratprof^*))},
$
where $\text{sat}(\stratprof)$ is the voters' satisfaction (measured either as SW or CC score) from the outcome elected by $f$ for $\stratprof$. Note that here, the positive value of a price of deviation indicates that a strategy profile provides higher satisfaction than  the natural profile $\stratprof^*$.

For the voting rules we consider, we test the social welfare in an equilibrium
and its corresponding price of deviation. Towards this end, we compute
equilibria for PSGs based on instances in Pabulib, analogously to the experiment
in~\Cref{sec:EquilibriumExistence}, and focus only on instances where
full strategy profile was not an equilibrium. Then, we consider the average
price of deviation for the computed profiles. The results of this experiment are
shown in~\Cref{table:SWResults}. 

Interestingly, for nearly all the considered rules and numbers of proposers, the average price of deviation measured by the CC score was negative. However, in case of SW, the results are more complex. We observe that in the case of \basicAV, the average price of deviation is negative, which means that reaching an equilibrium proved to be detrimental for the voters in the context of this rule. Nevertheless, in the case of \Phragmen, we observe a strong, opposite effect, showing that on average Nash equilibria lead to increased social welfare. In the case of \Mes, the effect is much smaller, although leaning towards positive price of deviation. Hence, we conclude that he plausible behavior of project proposers has a negative effect in the context of the commonly used BasicAV rule, which is in contrast with more proportional methods.

\begin{table}[t]
\caption{Average price of deviation. Each row
shows values corresponding to the price of deviation for particular numbers of
participating proposers.}\label{table:SWResults}
\centering%
\scalebox{.9}{\begin{tabular}{lcccccc}
\toprule
$\ell$ & \multicolumn{2}{c}{\Phragmen} & \multicolumn{2}{c}{\Mes} & \multicolumn{2}{c}{\basicAV}  \\
\midrule
   & CC & SW & CC & SW & CC & SW \\
\midrule
$2$ & $-0.02$ & $0.22$ & $-0.12$ & $0.06$ & $-0.18$ &$-0.03$ \\
$3$ & $0.03$ & $0.23$ & $-0.08$ & $0.04$  & $-0.17$ & $-0.08$ \\
$4$ & $0.04$ & $0.19$ & $-0.07$ & $0.13$  & $-0.18$ & $-0.03$ \\
$5$ & $0.02$ & $0.12$ & $-0.12$ & $-0.03$ & $-0.16$ & $-0.06$ \\
\bottomrule
\end{tabular}}
\end{table}

\subsection{Equilibria with Imperfect Information}\label{sec:experiments-limited-info}

While in earlier parts of our analysis we assumed that the proposers have
full knowledge about the preferences of the voters, it might appear that  strategizing under partial information can lead to much different outcomes. To challenge this intuition we empirically check whether choosing strategies based only on a small number of ballots yields a significantly worse utility for the proposers than in the situation in which voters' preferences are fully known.

In this experiment, we consider all of those 
PSG instances
analyzed previously in which there is an NE; as shown before, this is the case
for nearly all of the previously-considered instances. For each instance, we
execute the following scheme:

(1) For each fraction $\alpha \in \{ 1\%, 5\%, 10\%, 25\%, 50\%\}$, we sample
  $20$~subelections in which only preferences of the $\alpha$ of all voters are
  known. Subsequently, we compute Nash equilibria for games based on those
  subelections; we discard the rare instances with no NE. 

(2) For each subelection, we check the absolute difference between the utility of a
  proposer in the NE 
  for the subelection (from the perspective of the full PSG) and in the NE 
  for the full instance,
  and then normalize the difference by the latter utility. We record the maximum
  of such a normalized difference over all proposers.

Surprisingly, we observe that the differences in proposers' utilities are
limited. \Cref{table:basicAVDiffs} depicts the results for \basicAV. Indeed, as
the number of known ballots increases, the average maximum
normalized (absolute) differences drop below 20\%. Crucially, for a
rather small number of participants in PB elections, it is plausible for the
proposers to learn preferences of a high fraction of voters.

\begin{table}[t]
\caption{Imperfect-information effect for~\basicAV. Presented are the average
maximum normalized (absolute) differences of the utilities of proposers in an NE
for the full instance and for a subelection; shown separately for a range of 2
to~5~proposers.}\label{table:basicAVDiffs}
\centering%
\scalebox{.9}{\begin{tabular}{lccccc}
\toprule
$\ell$ & $1\%$ & $5\%$ & $10\%$ & $25\%$ & $50\%$ \\
\midrule
$2$ & $0.11$ & $0.10$ & $0.07$ & $0.05$ & $0.04$ \\
$3$ & $0.21$ & $0.20$ & $0.16$ & $0.14$ & $0.10$ \\
$4$ & $0.11$ & $0.09$ & $0.09$ & $0.07$ & $0.05$ \\
$5$ & $0.19$ & $0.19$ & $0.14$ & $0.10$ & $0.04$ \\
\bottomrule
\end{tabular}}
\end{table}

Those results are strengthened by the following complementary experiments.
First, as~\Cref{table:0DiffBasic} shows, for the case of \basicAV{}, 
in a vast majority of instances the utility of proposers was not altered at
all due to limited information.\footnote{In fact, in almost all of those cases proposers submitted all of the projects.} %
Second, we checked that playing according to recommendations obtained for partial knowledge is almost
universally better (or as good as) ignoring the issue and submitting all the projects
that a proposer has, even for 1\% samples.
While we focused on \basicAV{}, the results for the remaining rules presented
in~\Cref{app:omitted-experiments} are qualitatively very similar. Hence, we
argue that in the case of the games we consider, the equilibria computed under the
assumption of full knowledge remain close to those that can be achieved with
limited information.

\section{Related Work}\label{sec:related}
We conclude our discussion by describing related work.  Our study is
closely related to the line of work on \emph{strategic candidacy},
initiated by \citet{dutta2001strategic} and continued, e.g.,\ by
\citet{eraslan2004strategic},
\citet{rodriguez2006candidate,rodriguez2006correspondence},
\citet{lang2013new}, \citet{polukarov2015convergence}, and
\citet{brill2015strategic}. In a candidacy game, the candidates have
preferences on election outcomes and want to obtain as favorable
result as possible by choosing to either run or withdraw from the
race. The main issues regard existence, nature, and complexity of
computing Nash equilibria in such games. While most of this literature
focuses on the single-winner case,
\citet{obr-pol-elk-grz:c:multiwinner-candidacy-games} considered the
multiwinner setting. The main difference between PSGs and candidacy
games is that in the former the proposers derive value only from their
projects, whereas in the latter the candidates care about the
performance of the other candidates.

\begin{table}[t]
\caption{Percentage of instances for \basicAV{} with no differences between
maximum normalized (absolute) differences of the utilities in an NE for
subelections and that for the full instances, averaged over subelections; rounded to integers.}\label{table:0DiffBasic}
\centering
\scalebox{.9}{\begin{tabular}{lccccc}
\toprule
$\ell$ & $1\%$ & $5\%$ & $10\%$ & $25\%$ & $50\%$ \\
\midrule
$2$ & $47\%$ & $55\%$ & $62\%$ & $71\%$ & $77\%$ \\
$3$ & $51\%$ & $58\%$ & $64\%$ & $73\%$ & $82\%$ \\
$4$ & $58\%$ & $63\%$ & $70\%$ & $77\%$ & $85\%$ \\
$5$ & $66\%$ & $71\%$ & $80\%$ & $82\%$ & $87\%$ \\
\bottomrule
\end{tabular}}
\end{table}

Nominating candidates in an election was also studied
by \citet{faliszewski2016hard}, who asked for the complexity of
deciding if a party can nominate an election winner in a nonstrategic
setting. 
Their work focusing on the simple Plurality rule was later extended by~\citet{cechlarova2023hardness}
who studied more
involved rules. While the setting here is very related to ours, the
difference is that the authors focus on single-winner ordinal
elections, fix the submission limit, and ask for the possible
(or necessary) winners rather than for the existence of Nash
equilibria. Nonetheless, game-theoretic analysis of related settings
was covered, e.g., by \citet{harrenstein2022computing}.  Also,
the study of finding equilibria in nomination games has been explored
in several models related to voting, e.g.,
the Hotelling-Downs
model~\citep{harrenstein2021hotelling,deligkas2022parameterized,maass2023hotelling}
or
tournaments~\citep{lisowski2022equilibrium,lisowski2022strategic}.

In the context of participatory budgeting, our work is related to that
of \citet{sreedurga2023indivisible} and
\citet{fal-jan-kac-lis-sko-szu:c:pb-cost-games}, both of which consider the setting
where each project has several different variants, each with a
different cost. Further, shortlisting of projects has also been studied by \citeauthor{Rey2020ShortlistingRA}~[\citeyear{Rey2020ShortlistingRA}]. \citet{sreedurga2023indivisible} studies computational
and normative properties of voting rules that take different project
variants as explicit part of the input, whereas
\citet{fal-jan-kac-lis-sko-szu:c:pb-cost-games} consider Nash equilibria in games
where project proposers choose the costs (and want to maximize the
amount of funds they can get). These games are similar to ours for the
case of single project submission games, except that \citet{fal-jan-kac-lis-sko-szu:c:pb-cost-games}
assume that the costs come from a continuous interval and do not
affect voters' approvals, whereas we only consider a discrete set of
project variants, each with a possibly different support. Hence, our work concerns a significantly different type of proposers' actions. %

\section{Conclusion}\label{sec:summary}
We initiated the study of strategic behavior of project proposers in
participatory budgeting and multiwinner voting. While we showed
intractability of checking the existence of an NE in many settings, we also
found some tractable cases, of which the most important is that of
multiwinner voting, where each proposer nominates a single
candidate. Further, our experiments show that in practice
equilibria are common and often have a simple structure.

Our results open a large room for future work. In particular, it would be interesting to extend the study of study games where
proposers have limited knowledge about the votes and/or competing
proposers. While our experimental results make a step towards this goal, theoretical analysis of PSGs with limited information remains to be addressed, e.g., by incorporating uncertainty about voters' preferences or potential projects that other proposers can put forward. Furthermore, while pure strategy Nash equilibrium is an appealing concept, the study of PSGs would benefit from the study of more elaborate ones, such as trembling hand equilibria. Furthermore, it is of natural interest whether the intractability of computational problems we consider can be mitigated, e.g., through parametrized algorithms, with most natural being the number of candidates or the number of voters.

\section*{Acknowledgments}
We are grateful to AAMAS-26 anonymous reviewers for
their insightful feedback.
This
project has received funding from the European Research Council
(ERC) under the European Union’s Horizon 2020 research and innovation programme (grant agreement No 101002854).
Views and opinions expressed are however those of the authors only and do not necessarily reflect those of the European Union or the European Research Council. Neither the European Union nor the granting authority can be
held responsible for them.
This work was co-funded by the European Union under the project Robotics and advanced industrial production (reg. no. CZ.02.01.01/00/22\_008/0004590).
Grzegorz Lisowski acknowledges support by the European Union under the Horizon Europe project %
{Perycles}  (Participatory Democracy that Scales).
The work has been supported in part by the AI2050
program at Schmidt Sciences (Grant G-24-66104). Grzegorz Pierczy\'nski was supported by Poland’s National Science Center grant no. 2022/45/N/ST6/00271.

Andrzej Kaczmarczyk worked on the project in part
during his employment at AGH University and The
University of Chicago. Grzegorz Lisowski has been employed at AGH University while working on this paper.

\includegraphics[width=0.4\columnwidth]{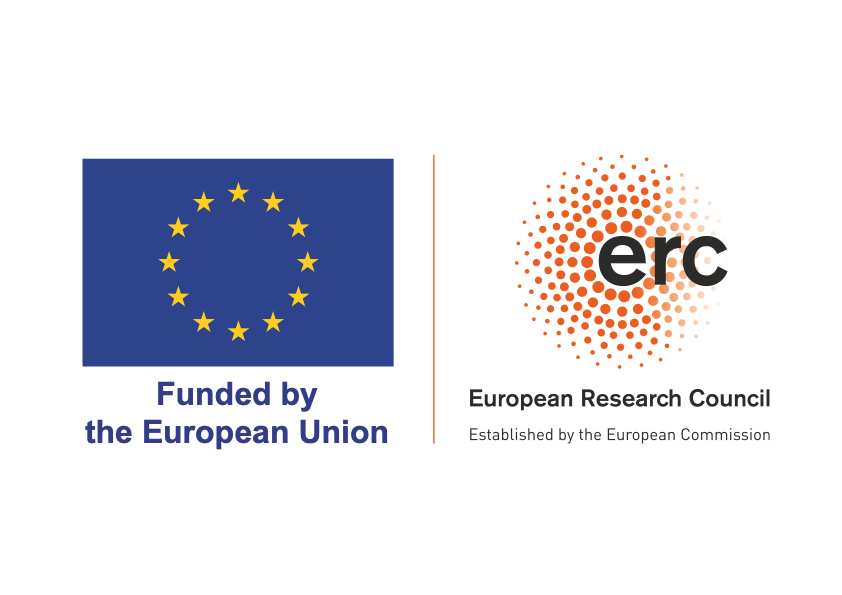}
~
\hspace{2em}
\raisebox{.8cm}{\includegraphics[width=0.4\columnwidth]{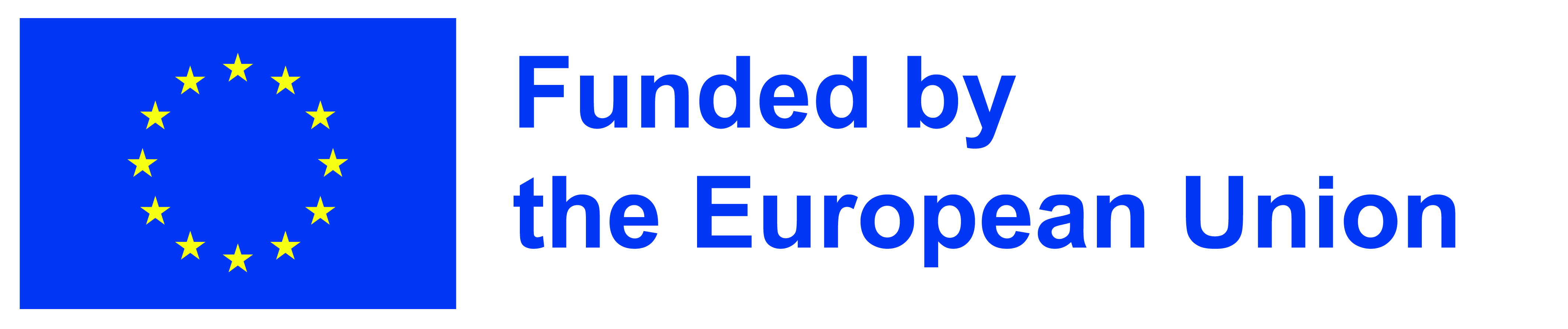}}

\bibliographystyle{ACM-Reference-Format} 
\balance
\bibliography{biblio}

\clearpage

\appendix

\section{Omitted Proofs}\label{app:proofs}

\thmnemainnegative*
\begin{proof}
    In the main text we have presented the reduction for $\np$-hardness for \basicAV\ and \Mes. Let us now show that the problem is also $\conp$-hard.

    Consider the same construction as for the $\np$-hardness case, with the following modifications: We increase the budget to $15\cdot T+14$ and add a new proposer $P_3$ with only one project costing $28$. This project is approved by everyone and in the tie-breaking is right after the projects of $P_0$ (and before the projects of $P_1, P_2$). Note that the only reasonable strategy (except for submitting an empty set) of $P_3$ is to submit their only project. Now if there exists a solution to the considered \textsc{Subset Sum} problem instance then the best strategy for $P_0$ is still to submit the projects corresponding to the choice of elements from $U$ summing up to $T$. After that, the project of $P_3$ is not funded, because it does not fit in the budget. Then there is no NE with respect to $P_1$ and $P_2$, as shown in the proof of \Cref{prop:NeNotExist}.
\end{proof}

\thmmultiwinnerhard*
\begin{proof}

 In the main text, we presented the NE existence proof for \basicAV{} and the reduction from 3-SAT to show that \Phragmen-\textsc{NE Existence} is $\conp$-hard. Here we will first describe the gadget used in the reduction in more detail, in particular, showing that there is indeed no NE in this game. Second, we will show how the presented 3-SAT reduction can be slightly modified to obtain constructions for \Mes\ and for $\np$-hardness.

 \textbf{Analysis of the gadget for \Phragmen.}\quad Let us show that the gadget PSG admits no NE. Suppose towards contradiction that such a profile $\stratprof$ exists. As noted in the main text, we can assume without loss of generality that projects $p_2, p_3, q_2, q_3, r_1, r_2$ are present in $\stratprof$, since their submission cannot lower the utility of $P_1, P_2, P_3$. 
 
 Let us consider the following, exhaustive cases.

\begin{enumerate}[label=(\arabic*)]
    \item Both $p_1, q_1$ are nominated. Then, observe that after purchasing $p_1$ and $q_1$, the voters supporting the remaining projects only have $\nicefrac{1}{2} - \nicefrac{1}{12} - \nicefrac{1}{10} = \nicefrac{19}{60} < \nicefrac{1}{3}$ budget left, and so no more projects are selected.
    \item Only one of the projects $p_1, q_1$ is nominated.  Then, we notice that after purchasing one of these projects, voters supporting $q_2$ and $q_3$ still have enough money to buy them. For voters supporting $p_1$ and $p_2$ this is not the case (note that these two projects can be only bought if their supporters do not spend their money on anything else).
    \item Neither of the projects $p_1, q_1$ is nominated. Then projects $q_2, q_3, p_1, p_2$ are bought. 
\end{enumerate}

The strategic game between $P_1$ and $P_2$ is summarized in \Cref{fig:mesnone}.
It is routine to check that there is no NE in this instance.

\begin{figure}[ht] 
		\begin{center}
\scalebox{1}{			\nfgame{$\emptyset$ $p_1$ $\emptyset$ $q_1$ $2$ $2$ $1$ $2$ $0$ $3$ $1$ $1$}}
		\end{center}
		\caption{Normal form representation, with rows representing $P_1$'s choices and column $P_2$'s. The values correspond to the number of members that proposers' have secured in the committee under each potential NE profile.}\label{fig:mesnone}
  \end{figure}

  On the other hand, if two voters approving $p_2$ and $p_3$
  respectively are removed (or their voting power is decreased so much
  that their presence does not matter, as it is the case in the
  presented reduction), then
  $(\{p_1, p_2, p_3\}, \{q_2, q_3\}, \{r_1, r_2\})$ is an NE
  profile. Indeed, then $p_1, q_2, q_3, r_1$ are purchased by
  \Phragmen. Player $P_1$ cannot improve their strategy, since it is
  not possible for them to have $p_2, p_3$ elected. Player $P_2$
  cannot improve their strategy, since if $q_1$ was submitted, it
  would be elected instead of $q_2, q_3$ (decreasing the utility of
  $P_2$).

 \textbf{$\boldsymbol \np$-hardness, \Phragmen.}\quad For an $\np$-hardness for \Phragmen, it is enough to take the same construction as for $\conp$-hardness with the following modifications: (1) remove $2$ voters approving $c^*$, (2) add two voters approving all the clause projects, (3) change the preferences of the two voters approving gadget projects and the clause projects, so that they approve the same gadget projects as before and $c^*$ instead. 

 Now by the very same reasoning as presented in the main text, we obtain that the gadget is autonomous (hence, there is no NE) if and only if $\varphi$ is not satisfiable.

\textbf{\Mes.}\quad The reductions for \Mes{} is analogous to the ones for \Phragmen, with the following modifications: (1) in the gadget, we have to change the budget from $4$ to $8$ (so that each group of $6$ voters corresponds to $\nicefrac{n}{\budget}$ voters), (2) in the whole construction, we assume that $w$ is divisible by $6$ (if not, we add up to $5$ new variables that do not appear in any clause) and set $\budget = 7/6w + 9$, (3) we remove the artificial projects and proposer $R$. Note that then in the whole construction we have that $\nicefrac{n}{\budget} = 6$. Then in the construction the behavior of \Mes{} is identical as \Phragmen{}, up to the fact that instead of electing some $x$ artificial projects, we elect a non-exhaustive outcome of $\budget-x$ projects. Hence, we can directly repeat the reasoning for the proofs for \Phragmen{} to prove both $\np$-hardness and $\conp$-hardness for \Mes{}.
\end{proof}

\thmbestresponse*
\begin{proof}

Points 2 and 4 follow directly from our NE constructions in \Cref{thm:multiwinner-hard} and \Cref{thm:ne-phragmen-mes-party-laminar} (in which players have dominant strategies). Point 1 follows from an analogous construction as in \Cref{thm:ne-main-negative} where we ask for a best response for player $P_0$. Let us now focus on the third point.

    \textbf{\Mes.}\quad We reduce from the NP-complete \xthreec{} problem, where for a universe set
	$\mathcal{U}=\{u_1, \dots, u_{3t} \}$ and a family~$\mathcal{S}=\{S_1, \dots,
	S_{3t} \}$ of three-element subsets of $\mathcal{U}$, we need to decide
	whether there are $t$~members of $\mathcal{S}$ that form a partition
	of~$\mathcal{U}$.

	\textit{Construction:} For an instance $I=(\mathcal{U},\mathcal{S})$ of \xthreec{}, we construct the
	corresponding game, which we call the \emph{encoding} of $I$. First, we let all of the projects in
	the cost 1 and $\budget = t$. Second, we build a set~$C= \{c(S)
	\mid S \in \mathcal{S}\}$ of projects. Then, we construct the
	set $V= \{v(u) \mid u \in \mathcal{U}\}$ of voters. We assume that for each $u\in \mathcal{U}, S\in \mathcal{S}$, voter $v(u)$ approves a candidate $c(S)$ if and only if $u\in S$. Finally, there is only one proposer, owning all the projects. Since there is only one proposer, the considered problem boils down to the question whether there exists a strategy giving $P_1$ utility $x$.
 
    \textit{Correctness:} Let $x=t$. We show that a strategy giving $P_1$ utility $t$ exists if and
	only if~$I$ is a positive instance of X3C.

	Suppose that there exists a set~$s_1 \subseteq \mathcal{S}$ being a partition of $\mathcal{U}$ and $|X|=t$, thus certifying that~$I$ is a positive instance of X3C. Consider an election~$E$ induced by player $P_1$ playing $s_1$. By the fact that~$X$ is a partition
	of~$\mathcal{U}$, every project in~$E$ is supported by its own group of
	three voters (that is, all three-supporter groups are pairwise
	non-intersecting).  In the \Mes{} procedure, each voter gets initially $\nicefrac{t}{3t} = \nicefrac{1}{3}$, and, consequently, each group of three voters is able to buy their approved project. Hence, all members of $s_1$ are elected and player $P_1$ gets utility $t$.

  Suppose now that $I$ is a negative instance of {\sc X3C}. It means that, no matter which strategy $P_1$ chose, there would be a voter approving no submitted project (and thus, not paying for any elected project). Since the total amount of money owned by the voters is $t$, this means that less than $t$ project would be bought.

  \textbf{\Phragmen.}\quad The reduction from X3C is analogous as in the case of \Mes, with the following modifications: we add $3$ additional "dummy" voters and an additional "dummy" candidate $c^*$ possessed by an additional proposer $P_2$. The project $c^*$ is the worst in the tie-breaking order. Here the only possible strategy of $P_2$ is to submit $c^*$, hence, all the considered problems are equivalent. Let us now show once again that proposer $P_1$ can gain utility $x=t$ if and only if $I$ is a positive instance of X3C.

  Suppose that there exists a set~$s_1 \subseteq \mathcal{S}$ being a partition of $\mathcal{U}$ and $|X|=t$, thus certifying that~$I$ is a positive instance of X3C. Consider an election~$E$ induced by player $P_1$ playing $s_1$ and player $P_2$ playing $c^*$. By the fact that~$X$ is a partition
	of~$\mathcal{U}$, every project in~$E$ is supported by its own group of
	three voters (that is, all three-supporter groups are pairwise
	non-intersecting). Note that under \Phragmen, each such group will buy their approved candidate at the moment $\nicefrac{1}{3}$ (when they collect $1$ dollar in total). At the same moment, dummy voters will be able to buy $c^*$ as well, yet $c^*$ will lose in the tie-breaking order with $t$ projects from $s_1$. Hence, player $P_1$ will get utility $t$.

  Suppose now that $I$ is a negative instance of {\sc X3C}. It means that, no matter which strategy $P_1$ chose, there would be a voter approving no submitted project. This means that at least one project of $P_1$ would be bought later than in the moment $\nicefrac{1}{3}$. But then project $c^*$ is bought by dummy voters at the moment $\nicefrac{1}{3}$ and $P_1$ gets utility only at most $t-1$.
\end{proof}

\propNeNotExistLimit*
\begin{proof}
    Take a PSG/1 with two proposers, $P_1 = \{a_1, a_2 \}$ and $P_2 = \{b_1, b_2 \}$. We assume that all of the projects are supported by all of the voters in the game. Let $\budget=6$. Also, let $ a_1 \succ b_1 \succ a_2 \succ b_2$. Finally, we let $a_1$ cost 1, $a_2$ cost 3, $b_1$ cost 4, and $b_2$ cost 5.

We observe that the utilities in the strategy profiles of this game are as summarized in the following table. Note that there is no NE in this game.

\begin{figure}[ht]\centering%
  \nfgame{$a_1$ $a_2$ $b_1$ $b_2$ $1$ $4$ $0$ $4$ $1$ $5$ $3$ $0$}
	\caption{Normal form representation, with rows representing the choices of $P_1$ and $P_2$.}\label{fig:MESParty-none-nfg}
\end{figure}
\end{proof}

\section{Thiele Rules}\label{app:thiele}

Besides the three rules studied in the main text, we also consider two types of Thiele
rules~\citep{thi:j:thiele,lac-sko:j:abc-rules}, which are frequently
studied in the context of multiwinner voting. Formally, for a multiwinner election $E = (C,V,\budget)$, they are
defined as follows:
\begin{description}%
\item[Sequential Thiele Rules.] We define a sequential $w$-Thiele rule
  $f$ using a non-increasing function
  $w: \mathbb{N}_+ \rightarrow \mathbb{R}_+$, such that $w(1)=1$,
  which we refer to as $f$'s weight function. The score of an outcome
  $W$ is:
  \[
    \textstyle w\hbox{-}\textnormal{score}_E (W)= \sum_{v\in V}\left( \sum_{\ell=1}^{|A(v) \cap W|}w(\ell) \right).
  \]
  A sequential $w$-Thiele rule selects a winning outcome~$W$ by
  starting with $W=\emptyset$ and performing $\budget$ iterations: In
  each iteration it selects a project $c \in C \setminus W$ that
  maximizes $w\hbox{-}\textnormal{score}_E (W \cup \{c\})$ and, then,
  includes $c$ in~$W$ (using the tie-breaking order if needed).  We
  particularly focus on two sequential Thiele rules, namely \basicAV{}
  (already defined in the main text) and sequential Chamberlin-Courant
  (SeqCC). BasicAV uses constant function $w_{\mathrm{av}}$ where
  $w_{\mathrm{av}}(i) = 1$ for each $i$, and SeqCC uses function
  $w_{\mathrm{CC}}$ such that $w_{\mathrm{cc}}(1)=1$ and
  $w_{\mathrm{cc}}(i)=0$ for all $i > 1$.
\item[Global Thiele Rules.] Given a weight function $w$ (defined
  above), a global $w$-Thiele rule selects the set $W$ of $\budget$
  projects that maximizes
  $w\hbox{-}\textnormal{score}_E (W)$.\footnote{In case of a tie, we
    use our tie-breaking order, extended lexicographically to apply
    to sets of projects. Namely, for outcomes $W_1$ and $W_2$, we have
    that $W_1 \succ W_2$ if there is a project $c \in W_1 \setminus W_2$
    that is preferred (according to $\succ$) to all the members of
    $W_2 \setminus W_1$.} Note that for $w_{\mathrm{av}}(i) = 1$ we
  obtain \basicAV{}. Global Thiele rule based on function
  $w_{\mathrm{cc}}$ is known as the (approval-based)
  Chamberlin--Courant (CC)
  rule~\citep{cha-cou:j:cc,pro-ros-zoh:j:proportional-representation,bet-sli-uhl:j:mon-cc}.
\end{description}

While sequential Thiele rules are computable in polynomial time
(assuming the weight function is easily computable), for global Thiele
rules the problem of deciding if there is an outcome with a certain
score is
intractable~\citep{sko-fal-lan:j:collective,azi-gas-gud-mac-mat-wal:c:approval-multiwinner}.
In principle, one could consider Thiele rules for the full-fledged
participatory budgeting setting, but we find it unnatural;
\citet{pet-pie-sko:c:pb-mes} offer a formal argument supporting this
view.

As we mentioned in the main text, Thiele rules behave very similarily to \Mes\ and \Phragmen. Specifically, for PSGs we obtain the completely analogous results, presented below.

\begin{restatable}{theorem}{thmnethiele}\label{thm:ne-thiele}
  In the multiwinner setting, \textsc{NE Existence} for all global and sequential Thiele rules except for \basicAV\ is both $\np$-hard and $\conp$-hard.
\end{restatable}
\begin{proof}
We will present reductions from 3-SAT, analogous as the ones for \Phragmen. There are only two differences: (1) we need to use a different gadget to get a PSG without NE, (2) we need to add several \emph{filler} projects to the construction to ensure that after electing some projects from the main construction, the ones with intersecting sets of supporters will be elected after the artificial projects. Formally, let the Thiele vector be $(1, 1, \ldots, 1, \alpha, \ldots)$ for some $0 \leq \alpha < 1$. Since the considered Thiele rule is not \basicAV{}, such $x \geq 2$ needs to exist. Denote the position of $\alpha$ in the vector by $x$. The number of filler projects that we need to add will correspond to $x$.

 \textit{The Gadget:} Consider an election with three proposers, $F=\{f_1, \ldots, f_{x-2}\}$, $P_1 = \{p_1, p_2, p_3, p_4 \}$ and $P_2 = \{q_1, q_2, q_3, q_4 \}$. We let the approval sets of projects in this
game be as depicted in \Cref{fig:ccbasicav-none}. We also assume the
tie-breaking order is $f_1 \succ \ldots \succ f_{x-2} \succ p_1 \succ q_1 \succ q_4 \succ p_3 \succ q_2 \succ q_3 \succ
p_2 \succ p_4$.
Finally, let each project cost one and $\budget=x+2$.

\begin{figure}[ht]
\centering
  \scalebox{0.6}{\begin{tikzpicture}
      
    \draw [fill=red!40!white] (1,0) rectangle (5,1);

    \draw [fill=yellow!40!white] (0,-1) rectangle (6,0);
    \draw [fill=yellow!40!white] (0,-3) rectangle (6,-2);
    \node[anchor=center, fontscale=2] at (3,-0.5) {\huge $f_{x-2}$};
    \node[anchor=center, fontscale=2] at (3,-1.5) {\huge $\ldots$};
    \node[anchor=center, fontscale=2] at (3,-2.5) {\huge $f_1$};

 \draw [fill=blue!40!white] (3,1) rectangle (6,2);

\draw [fill=blue!40!white](4,2) rectangle (5,3);
\draw [fill=blue!40!white] (5,0) rectangle (6,1);
\draw [fill=blue!40!white] (3,2) rectangle (4,3);

\draw [fill=red!40!white] (0,0) rectangle (1,1);
\draw [fill=red!40!white]  (1,1) rectangle (2,2);
\draw [fill=red!40!white]  (2,1) rectangle (3,2);

    \node[anchor=center, fontscale=2] at (0.5,-3.5) {\huge $6$};
    \node[anchor=center, fontscale=2] at (1.5,-3.5) {\huge $6$};
    \node[anchor=center, fontscale=2] at (2.5,-3.5) {\huge $6$};
    \node[anchor=center, fontscale=2] at (3.5,-3.5) {\huge $6$};
    \node[anchor=center, fontscale=2] at (4.5,-3.5) {\huge $6$};
    \node[anchor=center, fontscale=2] at (5.5,-3.5) {\huge $6$};

    \node[anchor=center, fontscale=2] at (0.5, 0.5) {\huge $p_2$};
    \node[anchor=center, fontscale=2] at (1.5,1.5) {\huge $p_3$};
    \node[anchor=center, fontscale=2] at (2.5,1.5) {\huge $p_4$};

    \node[anchor=center, fontscale=2] at (3.5,2.5) {\huge $q_2$};
    \node[anchor=center, fontscale=2] at (4.5,2.5) {\huge $q_3$};
    \node[anchor=center, fontscale=2] at (5.5,0.5) {\huge $q_4$};

     \node[anchor=center, fontscale=2] at (3,0.5) {\huge $p_1$};
    \node[anchor=center, fontscale=2] at (4.5,1.5) {\huge $q_1$};
  \end{tikzpicture}}
\caption{An illustration of the gadget. Projects of different
  proposers are depicted as boxes of different colors. Sets of voters
  of given cardinalities approve projects that are placed above
  them.}\label{fig:ccbasicav-none}
\end{figure}

Let us demonstrate that this game does not admit an sequential nor global CC-NE. Suppose that this is not the case and take an equilibrium profile $\stratprof$. First, we notice that in an NE player $F$ always submits all their filler projects and gets them elected. Moreover, submitting a project other than $p_1$ and $q_1$ does not lower the utility of their proposer. Hence, without loss of generality, we assume that those projects are submitted in $\stratprof$. Let us then consider the following, exhaustive cases. 
(1) $p_1, q_1$ are submitted. Then, we observe that sequential CC selects $p_1, q_1, p_2$, and $q_4$.
(2) $p_1$ is submitted but $q_1$ is not. Then, $p_1, q_4, p_2$, and $p_3$ are selected.
(3) $q_1$ is submitted but $p_1$ is not. Then, we have $q_1, p_2, p_3$, and $p_4$ in the committee.
(4) Neither $p_1$ nor $q_1$ are submitted. Then, $q_4, p_3, q_2$, and $q_4$ are selected.

The strategic game capturing the four discussed cases is shown
in~\Cref{fig:ccbasicav-none-nfg}. It is routine to check that there is no NE in
this game.

\begin{figure}[ht]\centering%
  \nfgame{$\emptyset$ $p_1$ $\emptyset$ $q_1$ $1$ $3$ $3$ $1$ $3$ $1$ $2$ $2$}
	\caption{Normal form representation, with rows representing $P_1$'s choices
	and column $P_2$'s. The values correspond to the number of members that
 proposers' have secured in the committee under each potential NE
  profile.}\label{fig:ccbasicav-none-nfg}
\end{figure}

\textit{Construction:} The constructions are the same as for \Phragmen, with the following modifications: (1) instead of $48$ gadget voters we have $36$ ones, (2) the two gadget voters approving either all clause projects (for $\conp$-hardness) or $c^*$ (for $\np$-hardness)  are now the two voters approving $p_3$ and $p_4$, respectively, (3), the budget is set to $\budget=x+w+4$, (4) proposer $R$ is additionally endowed with $x-1$ \emph{filler} projects that are first in the tie-breaking. The first $x-2$ of them are supported by all the voters; the last one is supported by all non-gadget voters.

\textit{Correctness:} First, in an NE, all $x-1$ filler projects are submitted and elected. Further, if $p_1, p_2, q_1$, or $q_4$ are submitted, they would be elected. Finally, the marginal contribution of all the other projects is equal ($6\cdot \alpha$) so the tie-breaking decides which projects are elected. However, similarly as in the case of \Phragmen, we get that out of these projects, no two ones with intersecting sets of supporters can be elected (in one of them is included, the marginal contribution of the other one would be smaller than $6\alpha$, i.e., smaller than the marginal contribution of $\budget-x+1$ artificial projects). Note that the above statement is true both for sequential and global (with lexicographical tie-breaking) Thiele. Then, by an analogous reasoning as in the $\conp$-hardness proof for \Phragmen{}, we obtain that the gadget is autonomous (hence, there is no NE) if and only if $\varphi$ is satisfiable. Note that if $p_5$ is "hit" and impossible to be elected, then there is an NE with respect to $P_1$ and $P_2$, where $P_1, P_2$ submit all their projects (and $p_1, p_2, q_1, q_4$ are elected). 
\end{proof}

\begin{theorem}\label{thm:thiele-partylist}
  Every global and sequential Thiele
for the multiwinner setting with party-list preferences, admits a Nash equilibrium.
\end{theorem}

\begin{proof}
    It is known that for party-list elections sequential and global Thiele rules are equivalent \cite{brill2018multiwinner}. Then the reasoning used for \Phragmen\ in the proof of \Cref{thm:ne-phragmen-mes-party-laminar} directly transfers to all Thiele rules.
\end{proof}

\begin{restatable}{theorem}{thmbestresponsethiele}\label{thm:nest-response-thiele}
  In the multiwinner setting, \textsc{Best Response} for all global and sequential Thiele rules except for \basicAV\ is $\np$-hard.
\end{restatable}
\begin{proof}
Let the Thiele vector be $(1, 1, \ldots, 1, \alpha, \ldots)$ for some $0 \leq \alpha < 1$. Since the considered Thiele rule is not \basicAV{}, such $x \geq 2$ needs to exist. Denote the position of $\alpha$ in the vector by $i$. 
  
  The reduction is the same as for \Phragmen, with the following change: we add proposer $P_3$ having $i-2$ \emph{filler} projects, approved by all the voters. Additionally, we increase the budget by $i-2$. Naturally, all these projects are submitted and elected primarily by both sequential and global Thiele. After this change, the further reasoning is identical as for \Phragmen.
\end{proof}

One could wonder whether in \Cref{thm:thiele-partylist} the requirement of party-list preferences is as strict as for the rules studied in the main text---specifically, whether it will not hold for laminar preferences. We certainly know from \Cref{thm:multiwinner-hard} that it is not the case for \basicAV. Moreover, the gadget presented in the proof of \Cref{thm:ne-thiele} is not laminar. It turns out that the requirement for party-list preferences is strict for nearly all Thiele rules, with only three exceptions.

\begin{restatable}{proposition}{propnethiele}\label{prop:ne-thiele}
For multiwinner model, there exists a laminar PSG with no NE for Thiele rules except for BasicAV, sequential CC and global CC.
\end{restatable}
\begin{proof}
Let us first prove two helpful lemmas. First of all, it turns out that the fact that sequential and global Thiele rules are equivalent for party-list preferences transfers also to laminar preferences.

\begin{restatable}{lemma}{propthielelaminar} \label{prop:thiele-laminar}
  For each weight function $w$, sequential and global $w$-Thiele rules
  are equal for laminar elections.
\end{restatable}
\begin{proof}
    Take any laminar election $E$, weight vector $w$, and a committee $W$ elected by Sequential $w$-Thiele. Let $W'$ be the committee maximizing $w$-score. Since $W$ and $W'$ have equal size and $W\neq W'$, there exists $c\in W\setminus W'$ and $c'\in W'\setminus W$. Denote by $V_c$, $V_{c'}$ the sets of supporters of these projects. Let us choose $c\in W\setminus W'$ so that %
    $c$ is the first project from $W\setminus W'$ selected by sequential $w$-Thiele.

    First, consider the $x^{\textit{th}}$ round in which sequential $w$-Thiele chose $c$. Denote by $W_x$ the committee elected by sequential $w$-Thiele before that round. It holds that:
    \begin{equation}
        \sum_{v\in V_c} w_{|A_v\cap W_x|+1} \geq  \sum_{v\in V_{c'}} w_{|A_v\cap W_x|+1}.
    \end{equation}
    Since the preferences are laminar and we assumed that $c$ was the first project from $W\setminus W'$ selected by sequential $w$-Thiele satisfying voters from $V_c$ (hence, $W_x\subseteq W\cap W'$), we know that in $W'$ there are no other projects approved by voters from $V_c$ than the ones from $W_x$. Indeed, if there was such a project $a$, then the sets of supporters of $a$ is either a strict subset of $V_c$, or it is equal to $V_c$ and $a$ is worse than $c$ in the tie-breaking order. In both cases, it would be better to swap $a$ with $c$ in $W'$. 
    
    Hence, the inequality is equivalent to: 
    \begin{equation}
        \sum_{v\in V_c} w_{|A_v\cap W'|+1} \geq  \sum_{v\in V_{c'}} w_{|A_v\cap W_x|+1}.
    \end{equation}
    On the other hand, we know that vector $w$ is non-increasing, hence for each voter $v\in V_{c'}$,  $w_{|A_v\cap W_x|+1} \geq w_{|A_v\cap W\cap W'|+1} \geq w_{|A_v\cap W'|}$. From that we obtain the following inequality:
    \begin{equation}
        \sum_{v\in V_c} w_{|A_v\cap W'|+1} \geq  \sum_{v\in V_{c'}} w_{|A_v\cap W'|}.
    \end{equation}

    However, this means that swapping $c'$ and $c$ in $W'$ would increase the $w$-score. Since $W'$ is the committee maximizing $w$-score, we obtain a contradiction.
 end{comment}
\end{proof}

From now on, we can then just focus on sequential Thiele rules. The second lemma describes the crucial property differentiating \basicAV\ and sequential CC from other sequential Thiele rules.

\begin{lemma}\label{lem:tripleValues}
For every sequential Thiele rule $f$ other than \basicAV{} and CC there are three subsequent values $\alpha, \beta, \gamma $ in the vector $w$ corresponding to $f$, such that $\alpha \geq \beta > \gamma$.
\end{lemma}

\begin{proof}
Take such a rule $f$ and its corresponding vector $w$. Also, let $d_w$ denote a sequence $(z_1, z_2, \dots )$, where $z_i$ is the set of all elements of $w$ of the same value. We assume that the values corresponding to elements of $z_w$ cover cover all values in $w$ and that $z_w$ is non-increasing.  Observe that since $f$ is not \basicAV{} or sequential CC, we either have that (1) $|z_1| >1$ and $|z_2|>0$, or (2) $w$ has at least three values.
If (1) is the case, then the claim holds immediately. If (2) is the case, then either some $|z_1|>1$ or for each $i$, $|z_i|=1$. In both cases the claim immediately holds.
\end{proof}

Now, take a sequential Thiele rule $f$ that is not \basicAV{} or CC, as
well as its corresponding vector $w$. Also, take values $\alpha, \beta, \gamma$
in $w$, as specified in \Cref{lem:tripleValues}. Finally, let $s$ be the position of $\alpha$ in $w$. 

Let us construct a game without an NE. First, take the same laminar example as for \Phragmen\ (gadget from \Cref{thm:multiwinner-hard}). Moreover, we take $s-1$ \emph{filler} proposers having one project each. Such projects are supported by all of the
voters. Also, let $\budget=s+3$.
Now, we notice that the actions of all proposers other than $P_1, P_2$ are
constant in all strategy profiles. Also, all projects of filler proposers are
always selected. Furthermore, as each voter receives less utility according to
$w$ for the third project in $P_1 \cup P_2$ than for the previous one, we
receive that the utilities of $P_1, P_2$ in the current game are the same as
in the proof of \Cref{thm:multiwinner-hard}. Hence, there is no $f$-NE in this PSG. 
\end{proof}

Regarding PSG/1s, \Cref{thm:pb-best-response} clearly transfers also to global and sequential Thiele rules (since still the number of strategies of each player is polynomial). What is less obvious, the same holds for \Cref{thm:ne-exists-for-limit-1}---the only difference is that for global Thiele rules, while an NE profile always exists, it is not polynomial-time computable. 

\begin{theorem}
    For all global and sequential Thiele rules, 
  every PSG/1 for the multiwinner setting has an
  NE profile. For sequential Thiele rules, it is polynomial-time computable.
\end{theorem}
\begin{proof}
For sequential Thiele rules we can use the same proof as for \Cref{thm:ne-exists-for-limit-1}. However, this reasoning works only for sequential rules, which means that it does not extend at all to global Thiele rules. Therefore, we present for them an alternative proof.

Consider a global Thiele rule $f$ for some PSG/1. Let us show
  that there exists an NE profile $\stratprof^*$ in this game. Let us
  first observe that each strategy profile $\stratprof$ corresponds to
  a set of projects selected by $f$. So, we can choose profile
  $\stratprof$ inducing an outcome maximizing the score under $f$,
  preferred among such outcomes in the tie-breaking order.
    
  We show that $\stratprof$ is an NE. Assume towards a contradiction
  that it is not. Then, take a proposer $P_i$ whose submitted project
  is not selected under $\stratprof$ and their representative $c_i'$
  who is selected under $(\stratprof_{-i}, c_i')$. Then the set of
  projects selected by $f$ under $\stratprof$ remains submitted in
  $(\stratprof_{-i}, c_i')$, but one of them is not selected in the
  changed profile.  This entails that the outcome of $f$ for
  $(\stratprof_{-i}, c_i')$ provides a better outcome than under
  $\stratprof$, which contradicts how~$\stratprof$ was formed.
    
\end{proof}

\begin{table}
	\begin{center}
  \caption{Equilibria-existence results for \Phragmen-PSGs (top) and
    \Phragmen-PSG/1s (bottom).  The column ``All'' shows the numbers of
    considered instances. For each instance, we recored if the respective full
    strategy profile was an equilibrium (``Full-NE''), if we obtained an
    equlibrium via the step of our algorithm that applies the best-response
    dynamics (``BR-NE''), if to find an NE we had to use the brute-force search
    (``BF-NE''), or if there was no equilibrium (``No NE'').}\label{table:exp-ne-exist-phragm}
\scalebox{1}{		\begin{tabular}{ c c c  c c c}
			
                  \multicolumn{6}{c}{\Phragmen-PSGs}\\[1mm]
                  \toprule
		$\ell$	& {All} & {Full-NE} &{BR-NE} & {BF-NE} & {No NE} \\ \midrule
		
		$2$  &  $187$ & $107$  &  $65$  &  $1$  &  $14$  \\ 
		$3$  &  $179$ & $132$  &  $35$  &  $3$  &  $9$  \\
	  $4$  &  $171$ & $141$  &  $28$  &  $0$  &  $2$ \\
		$5$  &  $140$ & $127$  &  $11$  &  $1$  &  $1$ \\
                  \bottomrule
                  \\[-2.5pt]
                  
                  \multicolumn{6}{c}{\Phragmen-PSG/1s}\\[1mm]
                  \toprule
        $\ell$	& {All} & {} & {BR-NE} & {BF-NE}  & {No NE} \\ \midrule
		
		$2$  &  $187$  &  & $176$  &  $0$  &  $11$  \\ 
		$3$  &  $179$  &  & $170 $  &  $0$  &  $9$ \\
	  $4$  &  $171$  &  & $166$  &  $0$  &  $5$ \\
    $5$  &  $140$  &  & $131$  &  $0$ & $9$ \\
                  \bottomrule
		\end{tabular}}
	\end{center}
\end{table}

\section{Alternative Definitions of \Mes}\label{app:mes}

We are aware of the fact that the general definition of \Mes{} proposed by \citet{pet-pie-sko:c:pb-mes} might at first glance seem to be a bit different from the definition we considered in the main text of the paper. Specifically, for general cardinal utilities (the setting where each voter $v$ has the utility function $u_v\colon C \to \naturals$ instead of the approval set $A_v$), the definition is as follows:

\begin{description}
    \item[\MesLong\ (\Mes).]%
  Initially, each voter gets $\nicefrac{\budget}{|V|}$ virtual
  currency. Then, each iteration proceeds as follows (let $b(v)$ be
  the money held by voter $v$ at the beginning of the iteration): For
  a project $c$ and a number $\rho$, we say that $c$ is
  $\rho$-affordable if it holds that:
  \begin{equation*}
    \textstyle
    \sum_{v\in N} \min(b(v), \rho \cdot u_v(c)) = \cost(c),
  \end{equation*}
  i.e.,\ if the voters approving $c$ can afford it, provided that each $v$
  of them contributes $\rho\cdot u_v(c)$ dollars (or all the money
  they have left if they have too little money).  {\Mes}
  includes in $W$ the project that is affordable for the lowest value
  of $\rho$ (and takes
  $\min(b(v), \rho\cdot u_v(c))$ currency from each $v \in S$; these voters \emph{buy} $c$).  {\Mes}
  terminates if no project is affordable for any $\rho$.
\end{description} 

Intuitively, value $\rho$ can be described as the \emph{maximal payment per point of utility}; in principle it means that voters should pay for projects proportionally to their utilities from them (e.g., a voter obtaining utility $10$ from a project $c$ should pay $5$ times more than a voter obtaining utility $2$ from $c$). The cost-utility variant of \Mes{} considered by us in the main text presumes that each voter $v$ has the following utility function: $\cost(c)$ for each $c\in A_v$ and $0$ for each $c\notin A_v$. This assumption is fairly standard, both in the literature \cite{pet-pie-sko:c:pb-mes,fal-fli-pet-pie-sko-sto-szu-tal:c:pabulib} and in practice (e.g., when \Mes{} was applied in Wieliczka \cite{boe:evaluation}). However, one could wonder, what would happen if we considered \Mes{} with a different approval-to-utilities mapping? %

Among such variants, the second most popular one in the literature (besides the cost-utility one) is the \emph{binary utilities} variant \cite{pet-pie-sko:c:pb-mes,fal-fli-pet-pie-sko-sto-szu-tal:c:pabulib}. Here it is assumed that each voter gains utility $1$ from each approved elected project and utility $0$ from each non-approved or unelected one. Formally, the main equation in the definition of \Mes{} looks now as follows:

\begin{equation*}
    \textstyle
    \sum_{v\in S(c)} \min(b(v), \rho) = \cost(c),
  \end{equation*}

\begin{table}
  \caption{Equilibria-existence results for \Mes-PSGs
    (top) and \Mes-PSG/1s
    (bottom). The column ``All'' shows the numbers of
    considered instances. For each instance, we recored if the
    respective full strategy profile was an equilibrium (``Full-NE''), if we
    obtained an equlibrium via the step of our algorithm that applies the
    best-response dynamics (``BR-NE''), or if there was no
    equilibrium (``No NE''). We never had to use the brute-force search for
    \Mes-PSGs or \Mes-PSG/1s.}
	\label{table:exp-ne-exist-mes}
  \centering%
		\begin{tabular}{ c c c  c c}
			
                  \multicolumn{5}{c}{\Mes-PSGs}\\[1mm]
                  \toprule
		$\ell$	& {All} & {Full-NE} & {BR-NE} & {No NE} \\ \midrule
		
		$2$ & $187$ & $167$ & $18$ & $2$  \\ 
		$3$ & $179$ & $170$ & $8$ & $1$  \\
	  $4$ & $171$ & $163$ & $8$ & $0$ \\
		$5$ & $140$ & $136$ & $4$ & $0$ \\
                  \bottomrule
                  \\[-2.5pt]
                  
                  \multicolumn{5}{c}{\Mes-PSG/1s}\\[1mm]
                  \toprule
        $\ell$	& {All} & {} & {BR-NE} & {No NE} \\ \midrule
		
		$2$ & $187$ & & $186$ & $1$ \\ 
		$3$ & $179$ & & $178$ & $1$ \\
	  $4$ & $171$ & & $170$ & $1$\\
    $5$ & $140$ & & $140$ & $0$\\
                  \bottomrule
		\end{tabular}
\end{table}

Obviously, in the multiwinner model, the utilities of all the voters from all the approved projects are the same. This means that all the results for \Mes{} for the multiwinner setting, presented in the main text, transfer directly to any other generalization of \Mes{} to PB. 

From the results for the general PB setting, only \Cref{thm:pb-best-response} transfers directly to the binary-utilities \Mes{} (with the same proof as presented in the main text). On the other hand, the constructions from \Cref{prop:NeNotExist} and \Cref{thm:ne-main-negative} do not transfer, since they both require the assumption that under equal support, all the projects are tied irrespectively of their cost (and more expensive ones are better in tie-breaking), while binary-utilities \Mes{} (similarly as \Phragmen{}) would prefer cheaper projects to more expensive ones in such a situation. However, the trick that worked for \Phragmen{} in the proof of \Cref{prop:NeNotExist} (constructing a plurality election in which each project has support equal to its cost) does not work for \Mes{}. Indeed, for each plurality \Mes{}--PSG there clearly exist an NE (each project can be then considered in isolation whether it is affordable or not). We leave the question whether there always exist an NE for party-list \Mes{}--PSGs (and \Mes{}--PSG/1s) in the general PB setting open for future research.

\section{Omitted Empirical Results}\label{app:omitted-experiments}

\Cref{table:exp-ne-exist-mes,table:exp-ne-exist-phragm} show the results  for
\Phragmen- and \Mes-PSGs and PSG/1s of the equilibiria-existence experiment
described in~\Cref{sec:EquilibriumExistence}. Overall the results are similar to
those for \basicAV{} and the differences are rather subtle: for \Phragmen, we
observe less frequently (yet still in the majority of cases) that full strategy
profiles are in NE. Here we also observe the only 5 games in which we could not
find NE profiles using out best response dynamics.  On the other hand, for
\Mes{} we observe that the equilibria are even more frequent and more
straightforward than in case of \basicAV.

We complement the results of the experiment
from~\Cref{sec:experiments-limited-info}, that is, we present the outcomes
for~\Phragmen{} and~\Mes{}.
Specifically~\Cref{table:PhragmenDiffs,table:MesDiffs} show qualitatively the
same results as those of~\basicAV{}, discussed in detail in the respective
section. Further, \Cref{table:0DiffPhragm,table:0DiffMes} show that our previous
observations for \basicAV{} hold also for PSG instances for~\Phragmen{}
and~\Mes{}. That is, a very significant percentage of instances, the equilibria
obtained under the assumption of imperfect information about the voters are of
the very same quality in terms of proposers' utilities.

\begin{table}[t]
\centering
\caption{Imperfect-information effect for~\Phragmen{}. Presented are the average
maximum normalized (absolute) differences of the utilities of proposers in an NE
for the full instance and for a subelection; shown separately for a range of 2
to~5~proposers.}\label{table:PhragmenDiffs}
\begin{tabular}{lccccc}
\toprule
$\ell$ & $1\%$ & $5\%$ & $10\%$ & $25\%$ & $50\%$ \\
\midrule
$2$ & $0.48$ & $0.29$ & $0.20$ & $0.13$ & $0.09$ \\
$3$ & $0.33$ & $0.26$ & $0.23$ & $0.20$ & $0.15$ \\
$4$ & $0.19$ & $0.11$ & $0.08$ & $0.06$ & $0.05$ \\
$5$ & $0.16$ & $0.10$ & $0.12$ & $0.06$ & $0.05$ \\
\bottomrule
\end{tabular}
\end{table}

\begin{table}[t]
\centering
\caption{Imperfect-information effect for~\Mes{}. Presented are the average
maximum normalized (absolute) differences of the utilities of proposers in an NE
for the full instance and for a subelection; shown separately for a range of 2
to~5~proposers.}\label{table:MesDiffs}
\begin{tabular}{lccccc}
\toprule
$\ell$ & $1\%$ & $5\%$ & $10\%$ & $25\%$ & $50\%$ \\
\midrule
$2$ & $0.08$ & $0.10$ & $0.07$ & $0.07$ & $0.07$ \\
$3$ & $0.04$ & $0.05$ & $0.05$ & $0.04$ & $0.02$ \\
$4$ & $0.03$ & $0.03$ & $0.03$ & $0.03$ & $0.02$ \\
$5$ & $0.03$ & $0.03$ & $0.03$ & $0.02$ & $0.02$ \\
\bottomrule
\end{tabular}
\end{table}

\begin{table}[H]
\caption{Percentage of instances for \Phragmen{} with no differences between
maximum normalized (absolute) differences of the utilities in an NE for
subelections and that for the full instances, averaged over subelections; the
entries are rounded to integers.}\label{table:0DiffPhragm}
\centering
\begin{tabular}{lccccc}
\toprule
$\ell$ & $1\%$ & $5\%$ & $10\%$ & $25\%$ & $50\%$ \\
\midrule
$2$ & $21\%$ & $41\%$ & $51\%$ & $63\%$ & $73\%$ \\
$3$ & $32\%$ & $47\%$ & $56\%$ & $66\%$ & $77\%$ \\
$4$ & $41\%$ & $52\%$ & $62\%$ & $72\%$ & $82\%$ \\
$5$ & $46\%$ & $58\%$ & $65\%$ & $76\%$ & $83\%$ \\
\bottomrule
\end{tabular}
\end{table}

\begin{table}[H]
\caption{Percentage of instances for \Mes{} with no differences between
maximum normalized (absolute) differences of the utilities in an NE for
subelections and that for the full instances, averaged over subelections; the
entries are rounded to integers.}\label{table:0DiffMes}
\centering
\begin{tabular}{lccccc}
\toprule
$\ell$ & $1\%$ & $5\%$ & $10\%$ & $25\%$ & $50\%$ \\
\midrule
$2$ & $55\%$ & $49\%$ & $56\%$ & $65\%$ & $73\%$ \\
$3$ & $65\%$ & $63\%$ & $60\%$ & $72\%$ & $79\%$ \\
$4$ & $68\%$ & $60\%$ & $69\%$ & $74\%$ & $82\%$ \\
$5$ & $74\%$ & $76\%$ & $76\%$ & $84\%$ & $88\%$ \\
\bottomrule
\end{tabular}
\end{table}

\end{document}